\documentclass[11pt]{article}
\usepackage{geometry}

\usepackage{graphicx} % Required for inserting images
\usepackage{booktabs}
\usepackage{placeins}
\usepackage{algorithm}
\usepackage{algpseudocode}
\usepackage{amsmath,amssymb, amsthm}
\usepackage{thm-restate}
\usepackage{amsfonts}
\usepackage{xcolor}
\usepackage[labelfont=bf]{caption}
\usepackage{subcaption}
 \usepackage[skip=5pt]{parskip}
 \usepackage{mdframed}
 \usepackage[colorlinks=true,pdfpagemode=UseNone,urlcolor=blue,linkcolor=blue,citecolor=violet,pdfstartview=FitH]{hyperref}
 \usepackage[noabbrev,nameinlink]{cleveref}
 \usepackage{fullpage}
 \usepackage{amsthm}

 \newtheorem{theorem}{Theorem}[section]
 \newtheorem{lemma}[theorem]{Lemma}
  \newtheorem{claim}[theorem]{Claim}
  
 \newtheorem{fact}[theorem]{Fact}
 \newtheorem{result}{Result}
 \newtheorem{remark}{Remark}
 \newtheorem{corollary}[theorem]{Corollary}

\newcommand{\eps}{\varepsilon}

\newcommand{\calA}{\mathcal{A}}

\newcommand{\calD}{\mathcal{D}}

\newcommand{\calP}{\mathcal{P}}

\newcommand{\Succ}{\mathsf{Succ}}

\usepackage[colorinlistoftodos,prependcaption,textsize=tiny]{todonotes}

\newtheorem{observation}[theorem]{Observation}
\newcommand{\Ot}{\ensuremath{\widetilde{O}}}

\DeclareMathOperator*{\Exp}{\ensuremath{{\mathbf{E}}}}

\renewcommand{\Pr}{\mathbb{P}}

\newcommand{\ceil}[1]{{\left\lceil{#1}\right\rceil}}

\newenvironment{rslt}{
	\begin{mdframed}[backgroundcolor=gray!20,topline=false,bottomline=false,leftline=false,rightline=false]
		\begin{result}
		}{
		\end{result}
	\end{mdframed}
}

\title{Randomized Algorithms for Learning Partitions with \\ Near Optimal Query Complexity in Constant Rounds\thanks{This work began at Dagstuhl Seminar 26131, ``New Trends in Clustering''. We thank the organizers for inviting us, and to the staff of Schloss Dagstuhl for providing a productive research environment.}} %with near Optimal Query Complexity}
\author{
	Deeparnab Chakrabarty\\
	Dartmouth College\\
	\texttt{deeparnab@dartmouth.edu}
	\and
	Aditi Dudeja\\
	The Chinese University of Hong Kong (Shenzhen)\\
	\texttt{aditidudeja@cuhk.edu.cn}
	\and
	David Saulpic\\
	CNRS and Université Paris Cité\\
	\texttt{david.saulpic@irif.fr}
}
\date{}

\begin{document}
\maketitle

\begin{abstract}
	We study the round complexity of learning a hidden partition $\mathcal{P}$ of an $n$-element universe using PAIR queries: PAIR($x,y$) 
	tells us whether $x$ and $y$ belong to the same part of the partition or not. While it is easy to learn using $n|\mathcal{P}|$ queries using a basic algorithm and this query complexity is optimal, 
	this basic algorithm is highly sequential. Black, Mazumdar, and Saha [COLT 2025] recently gave tight deterministic round/query tradeoffs when the number of parts of $\mathcal{P}$ is known. In particular they prove $\Theta(\log\log n)$ rounds are sufficient and necessary to limit the number of queries to $n|\mathcal{P}|$. They leave proving a randomized lower bound as an open direction. \smallskip
	 
	We show that randomization dramatically changes the picture. When the number of parts $k = |\mathcal{P}|$ is known, we give a simple 3-round randomized algorithm using $O(nk\log n)$ queries with high probability, and prove that 2 rounds require $\Omega(n^{4/3}k^{2/3})$ queries -- the same as deterministic algorithms. 
    We also study a more general setting where the number of parts is unknown. In this case, we give a 4-round randomized algorithm using $O(n|\mathcal P|\log^2 n)$ queries with high probability, and prove that 3-rounds cannot achieve near-optimal query complexity. 
	Furthermore, we show an even bigger separation in this regime between randomized and deterministic algorithms: for the latter, $\Theta(\log n/\log\log n)$ rounds are necessary and sufficient to 
	obtain near-optimal query complexity. \smallskip

        Thus, for partition learning, randomization yields a stark improvement in round complexity of algorithms with near optimal query complexity.

\end{abstract}
\thispagestyle{empty}
\newpage
\setcounter{page}{1}
\section{Introduction}

%what are you studying?
We consider the following learning problem in the query access model.
We are given a finite universe $U$ of $n$ elements. There is a hidden partition $\calP$ of $U$ that we have to learn. We have access to PAIR queries: for any pair 
of elements $\{x,y\} \subseteq U$, the PAIR query on $(x,y)$ reveals whether they are in the same part of $\calP$ or not.
It is easy to obtain an $O(n|\calP|)$-query algorithm by repeatedly picking an arbitrary unclassified element and performing (simultaneously) PAIR query with every other unclassified element
to learn the part containing that element. 
%where $k$ is the number of parts in $\calP$: we pick any element and make (at most) $(n-1)$ queries with everyone else
%and learn the full part containing this element; after that we remove and repeat $k-1$ times. 
Furthermore, it is known that this query complexity is optimal~\cite{DKMR14, MS17b, LM22}.
However, the above algorithm is highly {\em adaptive/sequential} and proceeds in $|\calP|-1$ {\em rounds} of question and answers, where the questions made in a round depends on the answers obtained in the previous rounds. A one round algorithm is also called a {\em non-adaptive} algorithm.
In this paper, we are interested in understanding what is the fewest rounds of queries any algorithm which maintains (nearly) optimal query complexity can learn in.
%the ``round-complexity'', that is the number of rounds of queries, of the above problem while maintaining (nearly) optimal query complexity.

%why is this important and interesting?
The above learning problem arises in many scenarios. One application is in {\em clustering} where the PAIR queries are known as ``same-set queries''~\cite{MT16,AKBD16, MS17a, MS17b}: the goal is to find a clustering of objects only based on questions to a user about pairs of objects. %without using any structural information about the cluster other than what the user wants -- and provides by answering aforementioned queries. 
This may be a clustering of pictures of animals, and the user may choose to cluster them according to the background color, the genus of the animals or their species, etc; the clustering algorithm would therefore only have access to  ``same-set queries'' provided by the user, and cannot assume any other structural information.
A closely related motivation comes from crowd-sourced entity resolution in databases, where the goal is to partition records
according to the real-world entities they represent.  In that setting, a natural
crowd question asks whether two records refer to the same entity, and many works~\cite{WLG13, WLKFF13, DKMR14, MS17c} 
study the question of how to select such crowd questions.
In both examples, query-complexity corresponds to the ``cost'' (in dollars paid to the user or the crowd) of the algorithm while the 
round-complexity corresponds to the ``execution time'' of the algorithm assuming a batch of questions are answered in a unit amount of time. It is therefore important to understand
the trade-offs between query and round complexity. %%%%deepc: should cite some similar considerations in other areas

%what was known, and what did you find out?
Black, Mazumdar, and Saha (abbreviated to BMS, henceforth) addressed the above question in their nice paper~\cite{BMS25} which gives the optimal answer for {\em deterministic} algorithms
with the assumption that the number of parts\footnote{$k$ could be an upper bound on $|\calP|$, but the query complexity is in terms of $k$ not $|\calP|$.} $k:= |\calP|$ is known to the algorithm;
we return to this assumption shortly.
First, for any $r \geq 1$, they describe a $r$-round deterministic algorithm to learn the partition which makes $8n^{1 + \frac{1}{2^r - 1}}k^{1 - \frac{1}{2^r - 1}}$ many queries. 
In particular, this implies\footnote{all logs are base 2} a $\ceil{\log\log n}$-round deterministic algorithm to learn a hidden partition of at most $k$ parts making $O(nk)$-queries. 
%At this point, we mention that their algorithm needs to know $k$, (an upper bound on) the number of parts in $\calP$; we return to this point shortly.
Second, they show that any $r$-round deterministic algorithm, even with the knowledge of the number of parts, $k$, 
must make $\Omega\left(\frac{1}{r}\cdot n^{1 + \frac{1}{2^r - 1}}k^{1 - \frac{1}{2^r - 1}}\right)$ queries. In particular, this rules out $\Ot(nk)$-query\footnote{As is usual, we use $\Ot(f(n))$ to denote $\leq f(n)\log^C f(n)$ for some constant $C$, for large enough $n$} algorithms in $< 0.9 \log\log n$ rounds. 
They remark that it is open to establish such a lower bound for arbitrary {\em randomized} algorithms. 

Our first result shows that there is no such lower bound, showing a stark contrast between deterministic and randomized algorithms.\smallskip

\begin{rslt}\label{res:1}
	We give a simple $3$-round randomized algorithm that given a parameter $k$, learns a hidden partition $\calP$ with at most $k$ parts making $O(nk\log n)$ many PAIR queries whp (see \Cref{thm:upper-bound-known-k}).
	
	Complementing the above, we show that any $2$-round randomized algorithm which succeeds with probability $\geq 0.51$ in learning a partition with $k$ parts, must make $\Omega(n^{4/3}k^{2/3})$ many queries (see \Cref{thm:lower-bound-known-k}). 
\end{rslt}

Note that BMS already implies a $2$-round deterministic algorithm with the query complexity of our lower bound. This shows that the ``power of randomization'' surprisingly kicks in at the third round, and not before that.
We also note that before our work,~\cite{BLMS24} showed that a non-adaptive (that is $1$ round) algorithm, deterministic or randomized, must make $\Omega(n^2)$ queries. 

%\deepc{should we describe the algorithm above. i am not sure if a separate subsection on ``our techniques'' is needed.}
%\david{since it's for SOSA, it may be good to do this here, to emphasize the algorithm is simple ?}

\paragraph{Unknown Number of Parts.}
As noted above, the BMS algorithm needs to know $k$, (an upper bound on) the number of parts of $|\calP|$, and so does our randomized algorithm alluded to in the above result. This raises the following question: what can one say about the round-complexity when the algorithm does not know anything about $|\calP|$? Do things drastically change? Apart from being an interesting scientific question, not knowing $|\calP|$ is in fact a natural assumption 
in many applications; for instance, in the clustering application the algorithm, or even the user, may not a priori know how many clusters there are.
% user(eg, the clustering one) and indeed many papers~\cite{MS, DKMR14} design algorithms assuming the same. 

Note that the trivial, $O(n|\calP|)$-query algorithm described in the first paragraph of this paper doesn't need to know $|\calP|$; so the query complexity is unchanged. However,
for algorithms proceeding in a small number of rounds, not knowing $|\calP|$ raises a tension: if our query complexity is supposed to be $\Ot(n|\calP|)$, then in any round, we cannot afford to make more than $\Ot(nk')$ queries, where $k'$ is the current {\em lower bound} on $|\calP|$ that our algorithm can deduce from the answers seen so far. In particular, the algorithm cannot make more than $\Ot(n)$-queries in the first round.

We show that for randomized algorithms one extra round suffices, and is indeed necessary, if we want near-optimal query complexity.
%\deepc{I wonder if stating the $n^{4/3}k^{2/3}$ in the result below actually, perhaps paradoxically, makes the impact of the statement weaker. This is due to the lack of matching upper bound.}
\medskip

%Indeed, it is not too hard to show (~\Cref{thm:lb-det-known-k}) that any deterministic algorithm for learning a partition $\calP$ which makes at most $\Ot(n|\calP|)$ PAIR queries for every $|\calP|$, 
%must proceed in $\Omega(\log n/\log\log n)$-rounds, and this is tight (see ??). In this unknown $k$ case, randomization shows its colors much more strongly. \deepc{this is wrong. error.}\smallskip

\begin{rslt}\label{res:2}
	We give a simple $4$-round randomized algorithm that, for any partition $\calP$, 
	learns it making $O(n|\calP|\log^2 n)$-many PAIR queries whp (see \Cref{thm:upper-bound-unknown-k}). The algorithm does not know anything about $|\calP|$.
	
	Complementing the above, we show that any $3$-round randomized algorithm that learns any partition $\calP$ with probability $\geq 9/10$, must make $\Omega(n^{4/3}|\calP|^{2/3})$-queries
	in some round with constant probability on some partition $|\calP|$ (see~\Cref{thm:lower-bound-unknown-k}). 
	This rules out a $3$-round randomized algorithm which learns any partition $\calP$ making $\Ot(n|\calP|)$-many PAIR queries whp. 
\end{rslt}
%\deepc{As I went over the technical section, I somehow felt that stressing ``can't estimate k in one round'' not that important to be made at this point. At some level, the point is a bit weak: sure we can't estimate k as in get a good approximation, but, as the technical portion mentions, we can estimate what we need to :-). So removing this point, but will keep it in~\Cref{sec:our-techniques} where it leads the story better.}
%We make a couple of remarks. First, one possible approach that the reader may be thinking of to 
%solve the unknown $k$ case may be the following: use the first round (which, as argued above, cannot make more than $\Ot(n)$ queries) to estimate $k$, and then apply the $3$-round algorithm for known $k$. Unfortunately, this is not possible --- as we prove in~\Cref{lem:lb-1-round-est-k}, any non-adaptive, possibly randomized, algorithm making 
%$\Ot(n)$ queries cannot distinguish between two instances, one which has $k = n^{0.51}$ and another with $k = n^{0.75}$. So, the algorithm needs to do something slightly different; we 
%delay this discussion to~\Cref{sec:our-techniques}. 

At this point, we also should remark about the round-complexity of {\em deterministic} algorithms for reconstructing partitions when $|\calP|$ is unknown and which have near optimal query complexity of $\Ot(n|\calP|)$.
Although deterministic algorithms are not the main focus of this paper, we discover an {\em exponential gap} in the round complexity of near optimal query complexity algorithm 
between the settings of known number of parts and unknown number of parts. This is in contrast to the randomized case where our results above imply a gap of one extra round.
While the results of BMS~\cite{BMS25} imply $\Theta(\log\log n)$ rounds are necessary and sufficient when $|\calP|$ is known, we show that when $|\calP|$ is not known, the round complexity is $\Theta(\log n/\log\log n)$. 
In particular, there is a simple algorithm (\Cref{thm:upper-bound-det-unknown-k}) making $\Ot(n|\calP|)$ queries running in $\ceil{\frac{\log n}{\log\log n}}$ rounds, and any algorithm having query complexity $n|\calP|\log^c n$ for any constant $c$ must proceed (\Cref{thm:lower-bound-det-unknown-k}) in $\Omega(\log n/\log\log n)$ rounds. This, and some more nuanced results, can be found in~\Cref{sec:det-alg-unknown-k}.

\paragraph{Subset Queries.}
BMS~\cite{BMS25} also consider the partition learning problem with the stronger ``subset queries'': given a subset $S\subseteq U$, a query $Q(S)$ returns the {\em number}
of parts the elements of $S$ lie in. More precisely, $Q(S) = \sum_{P\in \calP} \min(1, |P\cap S|)$. When $S$ is a pair of elements, $Q(S)$ is $1$ iff the PAIR query answers YES and is $2$ iff the PAIR query answers NO. 
The paper studies the power of such queries for ``small subsets'', in particular, when the subset $S$ needs to have $|S| \leq s$ for some parameter $s$.
Since $O(s^2)$ PAIR queries simulate a subset query of size $s$, the query complexity of the partition learning problem\footnote{It is also easy to show an $\Omega(n)$ lower bound for large $s$; so the lower bound should be thought of as $\max(n, nk/s^2)$.} is $\Omega(\frac{nk}{s^2})$. BMS give a {\em randomized} algorithm which learns the partition in $\log\log n$ rounds making $\Ot(\frac{nk}{s^2})$ $s$-bounded subset queries, while their above lower bound proves a matching round-complexity for {\em deterministic} algorithms. Our final result, which is an easy corollary of our results above and a non-adaptive randomized algorithm in~\cite{BMS25}, is that we can shrink the number of rounds to $3$ and $4$ in the known $|\calP|$ and unknown $|\calP|$ case, respectively. \smallskip

\begin{rslt}\label{res:3}
	When $|\calP|$ is known, 
	there is a $3$-round randomized algorithm that makes with high probability $\Ot(nk/s^2)$-many $s$-bounded subset queries and learns the partition $\calP$.
	When $|\calP|$ is unknown, there is a $4$-round randomized algorithm with the same guarantee.
\end{rslt}

\subsection{Our Techniques.}\label{sec:our-techniques}

\paragraph{Algorithms.} 
For our algorithms, the starting point is the ability of random sampling to build \emph{nets} -- a small subset of elements which hits any large enough set in a set system. Here, it means that a small random sample of elements will hit all large enough parts of the partition: formally, sampling $\Ot(T)$ points hits with high probability every part of size at least $n/T$. Therefore, if $k = |\calP|$ (or an upper bound on $|\calP|$) is known, then in one round one can sample $\approx \sqrt{nk}$ points and be guaranteed to contain an element from every part of size $\geq \sqrt{n/k}$. 
So, in round one, we perform all pairwise queries among these sampled elements to learn the local partitioning of this random subset. In the second round, using at most one representative from each part of this local partitioning, we can learn the complete parts in which these representatives lie with $\leq nk$ queries, thereby learning every part of size $\geq \sqrt{n/k}$. Since there are (at most) $k$ parts, the number of ``unclassified'' elements at this point is therefore $\leq k\cdot \sqrt{n/k} = \sqrt{nk}$. By performing a third round of all pair-wise queries between these unclassified points gives us the simple $3$-round $\Ot(nk)$ query algorithm to recover the partition. 

Note that the above algorithm needs to know $k$, and in particular if $k = \omega(\sqrt{n})$, it makes $\omega(n)$ queries in the first round. 
However, if we do not know $k$, then in the first round we cannot make more than $\Ot(n)$ queries since the number of parts could be $O(1)$.
At this point, one idea is to {\em estimate} $k$ in one round using $\Ot(n)$ queries, and then apply the above algorithm to get a $4$-round algorithm.
Unfortunately, with a budget of $\Ot(n)$ queries, one cannot 
estimate $k$ well in one round: any non-adaptive, possibly randomized, algorithm making 
$\Ot(n)$ queries cannot distinguish between two partitions, one with $k = n^{0.01}$ parts and another with $k = n^{0.5}$ parts (see~\Cref{sec:appendix}).
%\footnote{On the other hand, it is not too difficult to build a two-rounds algorithm that {\em effectively} estimates $k$; see~\Cref{sec:appendix}.}.

Nonetheless, we can obtain a $4$-round randomized algorithm by morally doing the above.
%
%
%The case for unknown $k$ is more intricate. The first algorithm idea would be to first estimate precisely $k$ using rounds with $\Ot(n)$ queries, and then use the $3$-round algorithm. However, estimating $k$ precisely enough can only be done in $2$ rounds: one can show that any non-adaptive, possibly randomized, algorithm making 
%$\Ot(n)$ queries cannot distinguish between two instances, one which has $k = n^{0.51}$ and another with $k = n^{0.75}$. On the other hand, it is not too difficult to build a two-rounds algorithm that estimates $k$. 
%
%
%As any deterministic algorithm requires $\Omega(\log n / \log \log n)$ rounds, this $5$ rounds already is a stark separation between randomized and deterministic. However, it is not tight: using a different strategy, we present a 4-rounds algorithm. 
The point is that we don't really need a good estimate of $k$; what we need a {\em lower bound} $k'$ on $k$ which may be quite far from the true value but still actionable in three subsequent rounds. Indeed, when using the $3$ round algorithm with a value $k'$, potentially  $\ll |\calP|$, the risk is that there remains many unclassified elements in the last round. All those items however are in small parts -- with size less than $\sqrt{n/k'}$, as we sketched above. Thus, our estimate $k'$ need only satisfy that only a few elements are in small parts, or more precisely, that at most $\sqrt{nk}$ elements, where $k$ is the true $|\calP|$, are in parts with size $\leq \sqrt{n/k'}$. We now show how to do this again by random sampling.
%This is indeed possible.

%Here is the main idea to obtain such a $k'$ in one round makin $\Ot(n)$-many PAIR queries.
The main observation is this. For any integer $s$, let $U_s$ be the elements contained in parts of size $[s,2s)$. If $|U_s| = \Omega(n/s)$, then note the following: (i) 	$k\geq \Omega(n/s^2)$, 
(ii) a random subset $R\subseteq U$ of size $O(s)$ hits  $U_s$, and (iii) for each $x\in R$, one can be reasonably sure $x\in U_s$ by making PAIR queries with $x$ and $O(n/s)$ random other elements and seeing $\Theta(1)$ successes.
The first round of our algorithm, for every $s$ which is a power of $2$, implements (ii) \& (iii) using $O(n)$-many PAIR queries, and notes $s^*$ to be the 
smallest $s$ for which we get a hit with $U_{s^*}$. Due to (i) above, we can use $k' = \Omega\left(\frac{n}{{s^*}^2}\right)$ as a lower bound. 
Furthermore, we can assert that the number of elements $P$ in parts of size $\leq s^* = \sqrt{n/k'}$ is at most $\Ot(\sqrt{nk})$, where $k$ is the true number of parts. To see this, note that for any $s < s^*$ we have $|U_s| = O(n/s)$ and by definition, we have $|U_s| \leq 2ks$. Multiplying and taking square-roots, we get $|U_s| \leq O(\sqrt{nk})$. Since $P$ is the union of at most $O(\log n)$ many $U_s$'s, we have $|P|\leq \Ot(\sqrt{nk})$, which is what we desired.

\paragraph{Lower Bounds.} %We show that even with the knowledge of $k$, one needs $3$-rounds. 
%\paragraph{$3$ rounds are necessary.} The number of queries is optimal, as shown in \cite{}. 
We show that even with knowledge of $k$, %$3$ rounds are necessary to attain $\Ot(nk)$ queries, with randomization -- while any deterministic 3-rounds algorithm must use $\Omega(n^{8/7}k^{6/7})$ queries. 
%In particular, we show that 
any $2$-round randomized algorithm must make $\Omega(n^{4/3}k^{2/3})$ queries matching the deterministic upper and lower bounds of BMS~\cite{BMS25}, hence formalizing that randomization only helps with three rounds. 

The main idea is inspired by the $r$-round deterministic lower bound in~\cite{BMS25}. The deterministic lower bound follows via an ``adversary argument'' where an adversary forms the partition on the fly responding to the queries made, trying to hid a ``key feature'' till requisite number of rounds go by.
For randomized lower bounds, we need to commit to a distribution of partitions; however, we can argue that only 2-rounds of queries  cannot reveal enough information to detect the ``key feature''.

To give more details, we construct a distribution over partitions, with either $2k+1$ or $2k+2$ parts, as follows. We first split the universe into $k$ roughly equal parts $U_1, \ldots, U_k$, and each $U_i$ identifies a suitably sized subset $R_i$, and forms parts $U_i\setminus R_i$ and $R_i$. It is instructive to think of every element $x\in U$ choosing $u_x \in [i]$ at random, and an independent $r_x \in \{0,1\}$ at random, and $x\in U_i$ if $u_x = i$ and $x\in R_i$ if $u_x = i, r_x = 1$.
Let $R = \cup_{i=1}^k R_i$. One then samples a pair $(a,b) \in R\times R$ and in the DOUBLETON world, we pluck $\{a,b\}$ out and form its own part and in the SINGLETON world we pluck it out and form two parts: $\{a\}$ and $\{b\}$. Any successful algorithm {\em must} query the pair $(a,b)$. This is the ``key feature'' alluded in the previous paragraph and this is the same as in~\cite{BMS25}.
The size of $|R|$ is so chosen that (a) the first round queries $Q_1$ have very few pairs with endpoints in the same $R_i$, and (b) nevertheless there is a lot of uncertainty among the choice of $(a,b)$. In particular, condition (a) allows one to formalize the idea that the answer to the first round queries doesn't really decrease the uncertainty induced by part (b). We leave the details to~\Cref{sec:lower-bound-known-k}.

When $k$ is {\em unknown}, we show that $3$-rounds do not suffice even with randomness. The lower-bound builds on that for the known-$k$ case, with a slight twist to force the first round to use at most $\Ot(n^{4/3})$ queries. The whole argument is then to show that such a small number of queries is basically useless to the algorithm.
To force a small number of queries in the first round, the hard distribution is built as follows: with probability $1/2$, it is the trivial partition with a single part. With the remaining probability, it is the hard instance for known $k$ described above, with $k = n^{0.67}$. In fact, we prove that any $3$-round algorithm must make $\Omega(n^{4/3}|\calP|^{2/3})$ many queries.

%
%In slightly more details, the instance is built as follows: each item $x$ samples a value $u_x \in \{1, ..., k\}$ and a $r_x \in \{0, 1\}$, with a tiny probability of $r_x = 1$. Last, two items $a$ and $b$ are selected uniformly at random. The clusters are the following: with probability $1/2$, $a$ and $b$ are together in a part $\{a, b\}$ (the DOUBLETON world) and probability $1/2$ are in singletons $\{a\}, \{b\}$ (SINGLETONS). The remaining clusters are $R_i$, namely the set of items with $u_x = i, r_x = 1$ and $U_i \setminus R_i$, those with $u_x = i$ and $r_i = 0$.
To prove the $3$-round lower bound, we show that 
%The difficulty of the lower-bound comes again from the dependency between the queries in different rounds. To handle this, we show that : 
the first round reveals information only about a tiny number of parts; and that after the second round, $a$ and $b$ are indistinguishable in a large set of items (just as in the case for known-$k$).
For the first round: we say that an item $x$ is ``revealed" in the first round either if there are more than $k/3$ queries of the form $(x,y)$ involving $x$, or if there is even one query $(x,y)$ wth $u_x = u_y = i$, that is, the endpoints are in the same $U_i$. A $i \in \{1, ..., k\}$ is revealed if there is one item $x$ revealed with $u_x = i$. We show that if $k = n^{0.67}$, then only $o(k)$ many $i$'s are revealed (as long as $|Q_1| = o(n^{4/3})$. For all revealed $i$'s, we disclose the $U_i$ to the algorithm.
 
Now, in the second round, the remaining randomness is the following: all values of $r$ are still fresh -- as items are revealed solely based on their value of $u$; and, as each item was compared against $k/3$ other items, there is still a lot of randomness in the choice of the $u$ values. 
This allows us to show that, even conditioned on the answer to the first round of queries, the set $R'$ of points $x$ such that $r_x = 1$ and all for all queries $(x,y)$, $u_x \neq u_y$ is still a $(1-o(1))$ fraction of the items with $r_x = 1$  -- both because of the remaining randomness on $u$'s, and because the algorithm mostly hits points with $r_x = 1$ with tiny probability. 
Further, $a$ and $b$ are in this set, and the argument about the third round mimics that of the second round in the known $k$ case.
%\textcolor{blue}{aditi: our deterministic lower bound is now more ``interesting" than before so it might be worth give a short overview here?}
%
% Hence, to distinguish between SINGLETONS and DOUBLETON in a last round, the algorithm must query all pairs in $R'$: this exceeds largely the query budget $O(nk)$. 
%The conclusion is therefore that any algorithm with optimal query complexity has to use at least $4$ rounds, which shows that our algorithm is tight.

\subsection{Related Works}\label{sec:related-work}
{\em Partition learning and round complexity.}
The works most relevant to ours is the line of works on learning partitions from pairwise oracle access.  Reyzin and Srivastava~\cite{RS07} studied the equivalent problem of learning connected components of a hidden graph from membership queries, and Liu and Mukherjee~\cite{LM22} proved tight adaptive query bounds for learning graph partitions.  
As mentioned earlier, Black, Mazumdar, and Saha~\cite{BMS25} initiate the systematic study of round complexity for this problem in the known $|\calP|$ setting; to our knowledge, our work
is the first to consider the round/query tradeoff question in the unknown $|\calP|$ setting. Interestingly, some papers, such as~\cite{MS17a,MS17b} do consider the learning problems in the unknown $|\calP|$ case, but they do not study the tradeoff question. 
The partition learning has also been studied in other query models. We already mentioned the subset-query algorithms~\cite{BLMS24} which, given a subset $S$, tells how many parts $S$ hits.
This is the weak subset query;\cite{BLMS24} also study the strong subset query model where one gets the sub-partition of $S$ as well. Interestingly, this doesn't give too much more power.
The weak subset query is also called a {\em rank}-query by~\cite{CL24} (since it is the rank of $S$ in the simple partition matroid induced by $\calP$); that paper studied the query complexity 
of recovering $\calP$ with no restriction on the size of the subsets, and gave an $O(n)$-query deterministic algorithm, and this query complexity is information-theoretically optimal. 
Finally, several works study partition or clustering recovery when same-cluster queries may be noisy or faulty~\cite{MS17a,MS17b,HMMP19,BCBLP20,DPMT22,DMT24}.  These results are complementary to ours: they address robustness to erroneous answers, while our oracle is noiseless and we focus on the tradeoff between query complexity and adaptive rounds.
%\smallskip

\noindent
{\em Broader query-access models.}
Partition learning is also part of a broader literature on reconstructing hidden combinatorial objects from indirect queries.  In graph and hypergraph reconstruction, additive queries ask for the number or total weight of hidden edges contained in a queried subset of vertices; this model was studied for graphs by Grebinski and Kucherov~\cite{GK00} and Mazzawi~\cite{M10}, and for hypergraphs by Bshouty and Mazzawi~\cite{BM10}. More recently, this question was studied~\cite{CL26} in the CUT query model as well, motivated by relation to submodular function minimization. 
Another well-studied model is the edge-detecting oracle, where a query set reveals whether it contains at least one hidden edge.  Alon et al.~\cite{ABKRS04} studied learning hidden matchings, while Angluin and Chen studied learning hidden graphs and hypergraphs~\cite{AC06,AC08}; subsequent work sharpened query and round tradeoffs for graph learning~\cite{AB19} and low-degree hypergraph learning, including hypermatchings~\cite{BHW22}, and specializations to Erdos-Renyi hypergraphs~\cite{ART25}.
These models are not identical to PAIR queries, but they share the same basic theme: understanding how query complexity and adaptivity interact when reconstructing an unknown discrete structure. \smallskip
%
%
%\noindentr
%{\em Motivation and applications.}
%As mentioned earlier, PAIR queries naturally arise in clustering and database applications.  
%In active clustering, such queries are usually called same-cluster or same-set queries: the
%algorithm may ask whether two objects belong to the same target cluster
%\cite{AKBD16,MT16,MS17a,MS17b,MP17}.  A closely related motivation comes from
%crowd-sourced entity resolution in databases, where the goal is to partition records
%according to the real-world entities they represent.  In that setting, a natural
%crowd question asks whether two records refer to the same entity.  Whang, Lofgren,
%and Garcia-Molina~\cite{WLG13} study how to select such crowd questions for entity
%resolution; Wang, Li, Kraska, Franklin, and Feng~\cite{WLKFF13} study crowdsourced
%joins while exploiting the transitivity of the matching relation; and Mazumdar and
%Saha~\cite{MS17c} give a theoretical analysis of heuristics for crowdsourcedr entity
%resolution.  These works motivate the two resources studied here: the number of
%queries corresponds to cost, while the number of adaptive rounds corresponds to
%latency.\smallskip

\subsection{Preliminaries}

The formal problem we consider is the following. 
An oracle commits to a partition $\calP$ of $\{1, ..., n\}$. Our algorithm is allowed to query the algorithm in adaptive rounds: at round $t$, the algorithm can select a set of PAIR queries $\{(x_i, y_i)\}$.
For query $(x_i, y_i)$, the oracle answers $Q(x_i, y_i) = $YES if $x_i, y_i$ are in the same part of the partition, NO otherwise.  

Note that the queries at round $t$ may depend on the queries and answers of previous rounds; however, the partition is fixed at the beginning, before any query are done. This partition may be chosen based on the distribution of queries, but not on the specific random choices -- i.e., the oracle knows the code of the algorithm, but not the random seed.

% In the problem with randomization we consider, we allow the algorithm to be randomized but the partition itself is fixed deterministically. More specifically, we assume that the oracle answering queries has access to the query distribution of the algorithm, and commits to a partition before any query are made by the algorithm. 

Because of the clustering motivation of our problem, we will interchangeably use ``parts" and ``clusters": the goal of the algorithm is to recover an unknown partition, or equivalently an unknown clustering.

\section{Algorithms for Pair Queries}

\subsection{Known Number of Parts}
We begin with proving the following.

\begin{theorem}\label{thm:upper-bound-known-k}
	Given an upper bound $k$ on the size of a partition $\calP$, there is a $3$-round randomized algorithm that learns $\calP$ and makes $O(nk\log n)$ PAIR queries whp.
\end{theorem}

The base of our algorithm is the following simple observation:
\begin{fact}\label{fact:random}
	Let $S \subset [n]$ be a (multi-)subset of $T>0$ random samples of  $U$, chosen independently and uniformly (with replacement).
	Let $P$ be a cluster of $\calP$ with size at least $n/T$. Then, $\Pr[S\cap P \neq \emptyset] > \frac{1}{2}$.
\end{fact}
\begin{proof}
	Consider the $T$ independent elements of $S$ (with possible repetition) in the order they are picked to be $s_1, \ldots, s_T$.
	For any $i$, the probability $\Pr[s_i \in P] \geq \frac{|P|}{n} \geq \frac{1}{T}$. Therefore, \[\Pr[S\cap P = \emptyset] = \Pr[\bigwedge_{i=1}^T \{s_i \notin P\}] \leq \left(1 - \frac{1}{T}\right)^T \leq \frac{1}{e} ~~\Rightarrow~~ \Pr[S\cap P \neq \emptyset] > \frac{1}{2} \qedhere\]
\end{proof}
%
%We then build on the algorithm and \cref{thm:bms} to generalize the result to weak subset queries, and show an algorithm with 3 rounds and $\Ot(\frac{nk}{s^2})$ weak subset queries.
%
%\subsection{Pairwise queries}
In the first round, the algorithm identifies one representative for each cluster with size at least $\sqrt{n\log n/k}$; call such clusters large. This is done by sampling $O(\sqrt{nk\log n})$ points and querying all pairs. The second round identifies all these large clusters, by querying the representatives and all other points. The last round queries all unclassified points by querying all unclassified pairs.

\begin{algorithm}
\caption{Learning $k$ Partitions}\label{alg:known}
\begin{algorithmic}
\State Let $R$ be a random subset of $[n]$ of size $10 \sqrt{nk \log n}$.
\State \textbf{Round 1}: For all pairs $x,y$ in $R$, query $Q(x,y)$.
\State Select a maximal set of points $S$ in $R$ that are in different clusters. 
\State \textbf{Round 2}: For all $x \in R$ and $y \in [n]$, query $Q(x, y)$.
\State For $y \in R$, denote $C_y$ the set of points in $y$'s cluster, as computed in round 2. 
\State Let $P' = [n] \setminus \cup_{y \in R} C_y$ be the not-classified points.
\State \textbf{Round 3}: Query $Q(x,y)$ for all pairs of points $x, y \in P'$.
% \State The queries of round 3 identify the clusters of all points in $P'$.
\end{algorithmic}
\end{algorithm}

First, we show in the next lemma that the algorithm is correct: 
\begin{lemma}\label{lem:known-correct}
    At the end of round $3$, the algorithm correctly knows all clusters with probability $1$.
\end{lemma}
\begin{proof}
    At the end of round $2$, the algorithm correctly identifies all clusters that intersect $R$. All remaining clusters consist of points in $P'$: they are therefore identified in round $3$.
\end{proof}

Now, we can bound the number of queries made by the algorithm. First, we show that the first two rounds classify most of the points:
\begin{lemma}\label{lem:rem_points}
Let $P'$ be as in the algorithm be the set of points not classified after round $2$. With high probability, $|P'| \leq \sqrt{nk \log n}$.
\end{lemma}
\begin{proof}
    Let $C \in \calP$ be a cluster with size at least $\sqrt{\frac{n\log n}{k}}$. 
    Partition $R$ into $10 \log n$ many sets of size $\sqrt{\frac{nk}{\log n}}$:  \Cref{fact:random} ensures that $C$ intersects each part with constant probability. Hence, $C$ intersects with at least one of them with high probability $1-1/n^{10}$: let $y$ be a point in this intersection. 

    At round $2$, the algorithms queries pairs $x, y$ for all $x \in [n]$: hence, the set $C_y$ is equal to $C$.

    Hence, with high probability, in $P'$ there are only points from clusters with size less than $\sqrt{n\log n/k}$. There are at most $k$ such clusters, so $|P'| \leq \sqrt{nk \log n}$. 
\end{proof}

\begin{lemma}
    With high probability, the number of queries made by~\Cref{alg:known} is $O(nk \log n)$.
\end{lemma}
\begin{proof}
    Round 1 uses $O\left((\sqrt{nk \log n})^2\right)$ queries. Round 2 identifies at most $k$ representatives, hence makes at most $nk$ queries. 
    \Cref{lem:rem_points} ensures that, with high probability, $|P'| \leq \sqrt{nk \log n}$: the number of queries made in round $3$ is therefore $O(nk\log n)$. Hence, in total there are $O(nk \log n)$ queries whp.
\end{proof}

\subsection{Unknown Number of Parts}
In this section, we consider the case when the number of parts is unknown. As mentioned in the introduction, BMS~\cite{BMS25} only consider the known $k$ case.
In~\Cref{sec:det-alg-unknown-k}, we consider the situation of deterministic algorithms and how the picture changes drastically when $|\calP|$ is unknown to the algorithm.
In this section, however, we 
%As stated in the introduction, when $k$ is unknown the round-complexity of query-optimal deterministic algorithms becomes $\Theta(\log n/\log\log n)$ rather than $\Theta(\log\log n)$ as proved by~\cite{BMS25}. We discuss this in~\Cref{sec:app-det-unknown}. \deepc{this is erroneous. needs to be revisited.}
%In this section, we 
focus on randomized algorithms and show that the round complexity increases by $1$. In particular, we
prove the following theorem.

\begin{theorem}\label{thm:upper-bound-unknown-k}
	There is a $4$-round randomized algorithm which learns $\calP$
	and with high probability makes $O(n|\calP|\log^2 n)$ many PAIR queries. This algorithm has no knowledge about $|\calP|$.
\end{theorem}

Recall from the introduction that in the first round the algorithm can make only $\Ot(n)$ queries since it doesn't know $n$.
Our algorithm computes a {\em lower bound} $k^*$ in the first round, and then simply applies the $3$-round algorithm (\Cref{alg:known}) for known $k$ in rounds 2,3, and 4.
Crucially, our lower bound $k^*$ may not be accurate, that is, it is possible that $k^* \ll k$. In such cases, however, we argue that the 
number of pairs queried in the last round of~\Cref{alg:known} (which is Round 4 of this algorithm) remains $\Ot(\sqrt{nk})$, for the true value of $k$, even though the two previous rounds used only $\Ot(nk^*) \ll nk$ queries.

The lower bound is based on the following idea. Suppose for some $s$, there are more than $n/(s \log n)$ points in clusters of size (roughly) $s$. In this case, a 
random subset $R$ of size $s\log n$ is likely to hit such a cluster. 
Fix some cluster $C$, and an item $x\in R \cap C$. Next, note that another sample $Q$
of  $n/s$ independently drawn points contains in expectation $|Q| \cdot \frac{|C|}{n} = \frac{|C|}{s}$ many ``neighbors'' (YES answer with a PAIR query) of $x$. 
Thus, if $C$ has size significantly less than $s$, $x$ has no neighbors in $Q$; if it has size roughly $s$, it contains a single neighbor with constant probability; and if it has size much bigger than $s$, it contains more than one neighbor.
Thus, by making PAIR queries between all points $s \log n$ points of $R$ and random samples $Q$ of size $n/s$ -- which is $n \log n$ queries-- it is possible to estimate how many points of $R$ are in clusters of size $s$; and, since $R$ is a random set, this allows us to estimate how many points in total are in clusters of size $s$.
% Thus, if $C$ has size (roughly) $s$, and say $x\in R \cap C$. Next, note that another sample $Q$
% of $n/s$ many independently drawn points contains in expectation $|Q| \cdot \frac{|C|}{n}$ many ``neighbors'' (YES answer with a PAIR query) of $x$ with constant probability, contains exactly one point of $C$. That is, $x$ has exactly one ``neighbor'' (YES answer with a PAIR query) 
% in $Q$. 
% So, it is possible with $s\log n$ sample to hit a cluster of size $s$, and a PAIR query for each sample to $n/s$ other points allow to identify whether 
Leveraging this idea, it is possible to identify $s^*$, the smallest $s$ such that there are more than $n/s$ points in clusters of size roughly $s$.
Since there are at most $k$ such clusters, this yields the bound $s^* k \geq \frac{n}{s^*}$ implying $k \geq \frac{n}{{s^*}^2}$. This is the ``lower bound'' $k^*$ that we use; we analyze in \Cref{cor:kstar-is-lb} that this is indeed a lower bound with high probability.

As alluded above, it is possible that $k^* \ll k$ in which case we lose our guarantee of~\Cref{alg:known}. In particular, the worry is that in the last round we may have too many not-classified points remaining since our choice of $k^*$ was small. We show this is however not a problem because the points left not-classified must be in clusters of some size $s$ such that fewer than $n/s$ points are in $s$-sized clusters (otherwise, our $k^*$ would've been higher). So if there are ``too many'' not-classified points remaining at the end of the penultimate round, then it must be because in reality there must be ``too many'' clusters as well. Thus, doing all pair queries is within budget. 

We now proceed to give details of the above idea; the formal algorithm is described in~\Cref{alg:unknown}. The high-level steps are the following.
%Hence, we can use $k^* := \frac{n}{{s^*}^2}$ as a lower bound for $k$; and using the same ideas as in the known-$k$ case, it is possible in two additional rounds to identify all clusters of size more than $s^*$, using $\Ot(nk^*) = \Ot(nk)$ queries. The number of points remaining may be upper-bounded by $k s^*$: however, this bound may be too weak -- it is not clear whether $(ks^*)^2 = \Ot(nk)$. Instead, we use the way $s^*$ is defined to show that, if $i$ is not successful, there cannot be many points in cluster of size $n/2^i$. This idea leads to our conclusion that the number of unclustered points is at most $\Ot(\sqrt{nk})$, and we can make all pairwise queries in the last round.

%The formal algorithm is described in \cref{alg:unknown}. The fact that it correctly classify all points is straightforward, as for the known-$k$ case. What is more difficult is to show that it does $\Ot(nk)$ queries. For this, we start by analyzing the guarantees of the first round, which proceeds as follows:
\begin{enumerate}
	\item Instead of working with every $s$, we work with powers of $2$.
	\item So, for each $i = 1, ..., \log n$, we think of $s = 2^i$ and sample $160 \cdot s \log n$ random points in $U$; call this set $R_i$.
	\item For each $x\in R_i$, we sample a set $Q_x$ of $n/s$ random points in $U$ and perform PAIR queries.
	\item We call $x$ {\em successful} if {\em exactly} one of the PAIR queries returns YES.
	% \item We call $i$ successful if any of the $x\in R_i$ are successful.
	\item We call $i$ {\em useful}  if at least $5\log n$ of the $x\in R_i$ are successful.
\end{enumerate}

\begin{algorithm}
	\caption{Learning $k$ Partitions with unknown $k$}\label{alg:unknown}
	\begin{algorithmic}
		\State \textbf{Round 1}. \Comment{We are being very loose with constant factors}
		\State $k^* \leftarrow 1$.
		\For{$i = 1, ..., \log n$}:
			\State Sample a set $R_i$ of $160 \cdot 2^i \log n$ random points from $U$ \Comment{$2^i$ is ``$s$'' in above discussion}
			\State Initialize $\Succ_i \leftarrow \emptyset$
			\For{$x\in R_i$}:
				\State Sample a set $Q_x$ of $n / 2^i$ random points from $U$ 
				\State Perform pair query $Q(x,y)$ for each $y\in Q_x$
				\If{exactly one query returns YES}: \Comment{this $x$ is {\em successful}}
					\State $\Succ_i \leftarrow \Succ_i \cup \{x\}$
				\EndIf
			\EndFor
			\If{$|\Succ_i| \geq 5\log n$}: \Comment{this is a {\em useful} $i$}
				\State Set $ k^* \leftarrow \frac{n}{4^{i}}\cdot \frac{1}{500 \log n}$.
				\State Break and exit the for-loop
			\EndIf
			%\State $i$ is successful is there is at least $5 \log n$ successful $x$ in $R_i$.
		% \State $i$ is successful is there is at least one successful $x \in R_i$.
		\EndFor
		%\State Set $i^* \leftarrow \min\left(\min_{i\in L} i, \log_2 \sqrt{n} \right)$.
		%\State Set $ k^* \leftarrow \frac{n}{2^{i^\star}}\cdot \frac{1}{2^{i^\star} \log n}$.
		\State {\bf Rounds 2, 3, 4}:
		\State Use~\Cref{alg:known} with the number of parts $\hat{k} \leftarrow 2000 k^*\log n$.
	\end{algorithmic}
\end{algorithm}

%\david{I don't like to have definition inside an algorithm, that's why I left this pseudo code -- but it is really redundant and the algorithm is simple, so I'm not sure how to do here}

\begin{lemma}\label{lem:unknown-facts}
	For any $x$, let $s(x)$ be the size of the cluster containing $x$. 
	For $i < \log (n/2)$ and $x \in R_i$, the probability that $x$ is successful, that is, in $\Succ_i$, is:
	\begin{itemize}
		\item at most $1/n^2$ if $s(x) > 3\cdot 2^{i} \log n$,
		\item at least $\frac{1}{16}$ if $s(x) \in [2^{i-1}, 2^i)$.
	\end{itemize}
	
	% In addition, if $u_i \geq n/2^i$, then there are $\frac{\log n}{e^2}$ successful $x \in R_i$ with probability at least $1-1/n^{10}$
\end{lemma}
\begin{proof}
	Fix such an $x$. For a random $y \in U$, the probability that the PAIR query between $x$ and $y$ returns YES is $\frac{s(x)}{n}$. Thus, as all the random samples are taken independently, %the expected number of YES between $x$ and $Q_x$ is $\frac{s(x)}{2^i}$, 
	the number of YES is Binomial variable with $\frac{s(x)}{n}$ success probability and $n_i := \frac{n}{2^i}$ random trials.
	The probability that there is exactly a single YES is therefore $n_i \cdot \frac{s(x)}{n} \cdot (1-\frac{s(x)}{n})^{n_i - 1}$ 
	
	In particular, if $n \geq s(x) > 3\cdot 2^{i} \log n$, this is at most 
	\[ n\left(1-\frac{s(x)}{n}\right)^{n_i - 1} \leq n\exp\left(- \frac{3\cdot  2^i \log n}{n} \cdot \frac{n}{2^i}\right) < n^{-2}.\]
	
	Similarly, if $s(x) \in [2^{i-1}, 2^i]$, then $1 \leq n_i \cdot \frac{s(x)}{n}$ and the probability is at least
	\begin{align*}
		\left(1-\frac{s(x)}{n}\right)^{n_i} &\geq \left(1-\frac{2^{i-1}}{n}\right)^{n/2^i}\\
		&\geq \frac{1}{16} 
	\end{align*}
	where we used that $x \rightarrow (1-x)^{2/x}$ is decreasing for $0 \leq x \leq 1/2$, and hence the term is greater than $1/16$ when $i \leq \log (n/2)$.
	%
	% For the last point of the statement, any $x \in R_i$ has probability $1/2^i$ to be in $U_i$: hence, since $|R_i| = 10 \cdot 2^i \log n$, the probability that none are in $U_i$ is $(1-1/2^i)^{|R_i|} \leq n^{-10}$.
\end{proof}

For $i = 1,2,\ldots, \log n$, define $U_i$ to be $x\in U$ with $s(x) \in [2^{i-1}, 2^i)$. Although we break the algorithm when we encounter any $i$ with $|\Succ_i| \geq 5\log n$, consider running 
the algorithm for all $i$ calling an $i$ {\em useful} if $|\Succ_i| \geq 5\log n$, and then setting $k^* := \frac{n}{4^i}\cdot \frac{1}{500 \log n}$ for the smallest useful $i$. If no $i$ is useful, then we set $k^* = 1$.

The above lemma implies the following as a corollary of~\Cref{lem:unknown-facts}.

\begin{lemma}\label{lem:useful-cor}
	With high probability, it holds that :
	\begin{itemize}
		\item if $|U_i| \geq n/2^i$, then $i$ is useful,
		\item if $i$ is useful, then $\left| \cup_{j : 2^j \leq 3\cdot 2^i\log n} U_j \right| \geq \frac{n}{160 \cdot 2^i}$.
	\end{itemize}
\end{lemma}
\begin{proof}
	First, fix $i$ such that $|U_i| \geq n / 2^i$. 
	Then, $|U_i \cap R_i|$ is a sum of $160 \cdot 2^i \log n$ independent Bernouilli trials, with success probability at least $|U_i|/n \geq 1/2^i$. Hence, Chernoff bounds ensure that with probability at least $1-1/n^3$, $|U_i \cap R_i| \geq 6 \cdot 16 \cdot \log n$; and \Cref{lem:unknown-facts} ensures that any point in this intersection will be successful with probability at least $1/16$. Since there are at least $6 \cdot 16 \cdot \log n$ such points, at least $5\log n$ of them are successful with high probability. Hence, $i$ is useful with high probability.

	For the converse, define $U_{\lesssim i} := \cup_{j : 2^j \leq 3\cdot 2^i\log n} U_j$ and suppose $|U_{\lesssim i}| \leq \frac{n}{160 \cdot 2^i} $.
	Then, $|U_{\lesssim i} \cap R_i|$ is a random variable with the same distribution as a sum of $160 \cdot 2^i \log n$ independent Bernouilli trials, each with success probability at most $\frac{1}{160 \cdot 2^i}$.
	Thus, Chernoff bounds ensure that with probability at least $1-1/n^2$, $|U_{\lesssim i} \cap R_i| < 3\log n$.
	In addition, a direct corollary of \Cref{lem:unknown-facts} is that, with probability at least $1-1/n^{2}$, no point in $R_i \setminus U_{\lesssim i}$ is successful (since they all have $s(x) \geq 3\cdot  2^i \log n$).  Hence, whp, there are at most $3\log n$ successful points in $R_i$, implying $i$ is not useful.
	%Therefore, if $|U_{\lesssim i}| \leq \frac{n}{160 \cdot 2^i \log n}$, then with high probability $i$ is not successful.
	Hence, by contraposition, whp, if $i$ is useful, then $|U_{\lesssim i}| \geq \frac{n}{160 \cdot 2^i}$.
\end{proof}

\begin{corollary}\label{cor:kstar-is-lb}
	With high probability, the $k^*$ returned by the algorithm satisfies $k^* \leq k$.
\end{corollary}
\begin{proof}
	This follows from the second bullet point of~\Cref{lem:useful-cor}. If the algorithm returns $k^* = \frac{n}{4^i}\cdot \frac{1}{500\log n}$, 
	then, this iteration $i$ is useful and so the lemma implies, with high probability, $\left| \cup_{j : 2^j \leq 3\cdot 2^i\log n} U_j \right| \geq \frac{n}{160\cdot 2^i}$.
	Since each part in this union is at most size $3\cdot 2^i \log n$, this implies $k \geq \frac{n}{160\cdot 2^i}\cdot \frac{1}{3\cdot 2^i \log n} > \frac{n}{4^i}\cdot \frac{1}{500\log n} = k^*$.
\end{proof}

\begin{theorem}\label{thm:unknown-queries}
	\Cref{alg:unknown} performs $O(nk\log^2 n)$ many pair queries with high probability.
\end{theorem}
\begin{proof}
	In the first round, the algorithm performs $\Ot(n)$ queries. In the second and third round, we run the first two rounds of~\Cref{alg:known}.
	So, the second round makes $O(n\hat{k}\log n)$ many queries which, by~\Cref{cor:kstar-is-lb}, is $O(nk\log^2 n)$ with high probability.
	The third round (i.e., second round of~\Cref{alg:known}) never makes more than $O(nk)$ queries. What is remaining and is non-trivial is the number of queries made
	in the last round -- this is $|P'|^2$ where $P'$ is the set of not-classified points. To analyze this, we will use the first bullet point of~\Cref{lem:useful-cor}.
	
	First, as in the analysis of~\Cref{alg:known}, we notice that with high probability $P' \subseteq U$ which lie in clusters of size $< \sqrt{\frac{n\log n}{\hat{k}}} = 2^{i-1}$, where $i$ 
	is the {\em useful} iteration at which we break. If there is no such $i$, then let $i = \log n$. So, $|P'| \leq \sum_{i' < i} |U_{i'}|$, or more pertinently, 
	there is an $i' < i$ with $\frac{|P'|}{\log n} \leq |U_{i'}|$ since there are only $\log n$ indices. Recalling the definition of the set $U_i$'s, we know that $k \geq |U_{i'}|/2^{i'}$. Putting all this together we get the following upper bound on the number of not-classified points.
	\begin{equation}\label{eq:ub-on-p'}
		|P'| \leq |U_{i'}|\cdot \log n \leq k\cdot 2^{i'}\cdot \log n
	\end{equation}
	
	Now, the first bullet point of~\Cref{lem:useful-cor} implies that, even after union bounding over all possible $i$'s, 
	with high probability $|U_{i'}| < \frac{n}{2^{i'}}$. Again using $|P'| \leq \frac{|U_{i'}|}{\log n}$, we get another upper bound
	\begin{equation}\label{eq:ub-on-p'-2}
		|P'| \leq |U_{i'}|\cdot \log n < \frac{n}{2^{i'}}\cdot  \log n
	\end{equation}
	Multiplying the above two inequalities gives that, with high probability, the number of queries made in the final round, that is, $|P'|^2$ is at most $nk\cdot \log^2 n$. This
	completes the proof of the theorem.
\end{proof}

\section{Lower Bounds for Pair Queries}
As usual, we will use Yao's lemma to focus on deterministic algorithms which are promised a partition from a particular distribution. Therefore, our lower-bound instances consists of distribution against which any deterministic algorithm must do many PAIR queries.

\subsection{2-round lower bound for arbitrary known $k$}\label{sec:lower-bound-known-k}

Let $k$ be any integer dictating the upper bound on the number of parts, and define $n=kl$ for some large enough integer $l \geq 1$. Let $U$ be a universe on $n$ elements.
We define a distribution $\calD$ over partitions of $U$ with at most $2k+2$ parts\footnote{Note that it is natural that the lower bound holds for partitions with at least $3$ parts: for $2$ parts, the trivial deterministic algorithm described in introduction uses only $1$ rounds.} 
We show that any 2-round deterministic algorithm $\calA$ which makes\footnote{For reading convenience, we use the little-oh notation throughout. One should should think of this as $\leq n^{4/3}k^{2/3}/C$ for some large constant $C$ which will determine the mistake probability.} $o(n^{4/3}k^{2/3})$ 
%$\leq \frac{n^{4/3}k^{2/3}}{1000}$ 
queries will make a mistake with at least a constant probability. In a subsequent remark we see how to boost this probability.
The distribution is defined as follows.

\begin{enumerate}
    \item[(D1.)] For each $x\in U$ we pick $u_x$ uniformly at random from $[k]$, let $U_{i}=\{x\mid u_x=i\}$.
    \item[(D2.)] For each $x\in U$ let $r_x=1$ with probability $\frac{1}{l^{1/3}}$ and $R_{i}=\{x\mid u_x=i, r_x=1\}$ and $R=\{x\mid r_x=1\}$.
    \item[(D3.)] Sample $a,b$ uniformly at random from $R$.
    % (or think of this as follows: each vertex $x\in R$ has random variables $a_x,b_x$ associated with them, we pick $a,b$ and set the corresponding random variables to $1$; the other are 0.)
    %deepc: was this $a_x$ and $b_x$ really necessary?
    \item[(D4.)] 
    Now we are ready to return the partitions.
    
    With probability $0.5$, we return the partition:  $U_1\setminus R_1,\cdots, U_{k}\setminus R_k,R_1\setminus \{a,b\},\cdots, R_k\setminus \{a,b\},\{a,b\}$. We call this the DOUBLETON instance
    
    With probability $0.5$, we return the partition $U_1\setminus R_1,\cdots, U_{k}\setminus R_k,R_1\setminus \{a,b\},\cdots, R_k\setminus \{a,b\},\{a\},\{b\}$. We call this the SINGLETON instance.
    
%    The clustering is defined as follows. 
%    \begin{enumerate}
%        \item With probability $0.5$, it is
%        $U_1\setminus R_1,\cdots, U_{k}\setminus R_k,R_1\setminus \{a,b\},\cdots, R_k\setminus \{a,b\},\{a,b\}$, which we call DOUBLETON. With probability $0.5$, it is $U_1\setminus R_1,\cdots, U_{k}\setminus R_k,R_1\setminus \{a,b\},\cdots, R_k\setminus \{a,b\},\{a\},\{b\}$, which we call SINGLETON.
%    \end{enumerate}
\end{enumerate}
    We will show that, with constant probability, $\calA$ cannot distinguish between a DOUBLETON or SINGLETON instance, and therefore, will make a mistake.

%\begin{remark}
%	Need to say a bit about exactly $k$ and at most $k$. {\bf to do}
%\end{remark}

Let us fix a 2-round deterministic algorithm $\calA$ which makes at most $o(n^{4/3}k^{2/3})$
%$\frac{n^{4/3}k^{2/3}}{1000}$ 
queries in each round. 
%The following theorem implies our lower bound via Yao's lemma.
\begin{theorem}\label{thm:det}
	Given a partition drawn from $\calD$, $\calA$ cannot distinguish whether it is a DOUBLETON instance or a SINGLETON instance 
	with probability $> \frac{1}{2} + o(1)$.
\end{theorem}

Using Yao's lemma, the above proves the following theorem alluded to in~\Cref{res:1}.
\begin{theorem}\label{thm:lower-bound-known-k}
	Any $2$-round randomized algorithm which learns a partition $\calP$ with at most $k$ parts with probability $> 0.51$
	must make $\Omega(n^{4/3}k^{2/3})$ many PAIR queries.
\end{theorem}

\begin{remark} \label{rem:lb-2rnd}
	%The above theorem, along with Yao's lemma, shows that any $2$-round randomized algorithm which makes $o(n^{4/3}k^{2/3})$ queries cannot recover the partition with probability $0.51$. 
	One can ``boost'' the $0.51$ success (or rather lack thereof) to any constant $\delta > 0$
	by taking $t = \Theta(\log(1/\delta))$ copies while increasing the universe size by a factor $t$. In the remainder of the section, we
	focus on proving~\Cref{thm:det}.
\end{remark}

%In the remainder of the section, we prove the above theorem.
We let $Q_1$ denote the first round queries, and let $Q_2$ denote the second round queries.
The key observation is that if $(a,b) \notin Q_1 \cup Q_2$, then the algorithm cannot distinguish a DOUBLETON instance from a SINGLETON instance, since one can form a coupling
among these with $(a,b)$ being the only query at which their answers differ. So if we can upper-bound the probability of $\calA$ querying $(a,b)$ by $o(1)$, we obtain the above theorem.
Here, the probability is over the randomness in $u_x, r_x, a_x, b_x$ for all $x\in U$.

By symmetry, it is easy to see that $\Pr[(a,b) \in Q_1] = \frac{|Q_1|}{\binom{n}{2}}$ since $\{a,b\}$ is uniformly distributed over all pairs.
Since $|Q_1| = o(n^2)$, this probability is $o(1)$. What is non-trivial is to show that the answers to $Q_1$, which we denote $A(Q_1)$, don't reveal enough information so that $Q_2$ can ``catch'' $(a,b)$. 
More precisely, we need to show
\begin{equation}\label{eq:nts}
\Pr[(a,b) \in Q_2~|~A(Q_1)] = o(1) \tag{*}
\end{equation}
This is what we do next.

We begin with a simple structural fact showing that $|R|$ is large with high probability.

\begin{fact}\label{fact:size_R}
	Let $G_1$ be the event that $|R| = (1-o(1)) \cdot k l^{2/3}$. Then $\Pr[G_1] = 1 - o(1)$
	%Let $G_1$ be the event that $|R| \ge 9k l^{2/3} / 10$. $\Pr[G_1] \geq 0.99$
\end{fact}
\begin{remark}
	The above is an example of an abuse of notation which we choose to do for the sake of readability (and not carrying around inconvenient constants).
The above statement should be read is as follows: for any constants $\eps_1, \eps_2 \in (0,1)$, 
	for large enough $n$, we have $\Pr[|R| \geq (1-\eps_1)k\ell^{2/3}] \geq 1 - \eps_2$.
\end{remark}

\begin{proof}
The expected size of $R$ is $\Exp[|R|] = \frac{n}{l^{1/3}} = kl^{2/3}$. The above fact now holds by Chernoff bounds.
More precisely, for any $\eps_1 \in (0,1)$, we have $\Pr[|R| < (1-\eps_1)\Exp[|R|]] \leq \exp(-\frac{\eps_1^2}{2}\cdot kl^{2/3})$, which can be made $\leq \eps_2$
for any $\eps_2 \in (0,1)$ by choosing $l$ to be large enough.
\end{proof}
Next, we state a few claims about the first round queries $Q_1$. First, note that $|Q_1| = o(n^{4/3}k^{2/3}) = o(k^2\ell^{4/3})$.
%$|Q_1| \leq \frac{n^{4/3}k^{2/3}}{1000} = \frac{k^2\ell^{4/3}}{1000}$.
We next show that with high probability, a good fraction of the queries have their endpoint lying in different $R_i$'s.

\begin{fact}
	Let $G_2$ be the event that $|Q_1\cap \cup_{i=1}^k(R_i\times R_i)| = o(k l^{2/3})$. Then, $\Pr[G_2] = 1 - o(1)$.
	%It holds that $\mathbb{E}[|Q_1\cap \cup_{i=1}^k(R_i\times R_i)|] \leq \frac{k l^{2/3}}{1000}$. Hence, with probability $99/100$, $|Q_1\cap \cup_{i=1}^k(R_i\times R_i)| \leq k l^{2/3}/10$.  We call this event $G_2$.
\end{fact}
\begin{proof}
	Let $(x,y)\in Q_1$ be an arbitrary query in $Q_1$. Let $u_x = i$. The probability that $(x,y) \in \cup_{i=1}^k (R_i\times R_i)$ is
	the probability $u_y = i, r_x = 1, r_y  =1$. Since these are independent, this probability is $\frac{1}{kl^{2/3}}$. 
	Therefore, $\Exp[|Q_1\cap \cup_{i=1}^k(R_i\times R_i)|] = o(k l^{2/3})$. The fact follows from Markov's inequality.
	\footnote{
	As before, to be precise, if $|Q_1| = n^{4/3}k^{2/3}/C$, then Markov's inequality states that 
	for any $\eps > 0$, we have $\Pr[|Q_1\cap \cup_{i=1}^k(R_i\times R_i)| \leq \frac{1}{\eps}\cdot \frac{kl^{2/3}}{C}] \geq 1 - \eps$.
	For any slowly growing function $C$, if we choose $1/\eps$ to be slower function such that $C\eps$ goes to infinity as well.
	Thus, quantitatively, the ``little-oh''s differ, but figuring this out exactly, in our opinion, is not worthwhile.}
	%The upper bound on the expectation follows from the upper bound on $|Q_1|$, and the probability statement follows from Markov's inequality. 
\end{proof}

Consider a pair $(x,y)\in Q_1$. If $u_x \neq u_y$, the answer to this query is NO. 
Crucially, if $x\in R$, then the answer to this query doesn't give any information of whether $x$ is $a$ or not. 
To this end, define
\[
R' :=\{x\in R\mid \text{for all}~(x,y)\in Q_1, u_x\neq u_y\}
\]

\begin{lemma}\label{lem:b3}
	Let $G_3$ be the event that the pair $(a,b)$ sampled in step (D3.) lies in $(R'\times R')\setminus Q_1$. 
	Then $\Pr[G_3] = 1-o(1)$.
	%$\Pr[G_3] \geq 0.75$. 
	% Then, $\Pr(G_2, G_3) = 1-o(1)$.
\end{lemma}

\begin{proof}
	Conditioning on $G_1$ and $G_2$, we can lower bound the size of $|R'|$: fix any item $x \in R \setminus R'$, with $u_x = i$. Then there is some $y$ with $u_y = i$, such that $(x,y) \in Q_1 \cap (R_i \times R_i)$. 
	Hence, $|R \setminus R'| \leq |Q_1 \cap \cup_{i=1}^k (R_i \times R_i)| = o(kl^{2/3})$.~\Cref{fact:size_R} implies $|R| = (1- o(1))kl^{2/3}$ implying,
	%Hence, $|R \setminus R'| \leq |Q_1 \cap \cup_{i=1}^k (R_i \times R_i)| \leq kl^{2/3}/10$.~\Cref{fact:size_R} implies $|R|\geq 9kl^{2/3}/10$ implying,
	%$|R\setminus R'| \leq |R|/9$, or $|R'| \geq 8|R|/9$. 
	$|R\setminus R'| = o(|R|)$ or $|R'| = (1 - o(1))|R'|$. %\leq |R|/9$, or $|R'| \geq 8|R|/9$. 
%
%	First, since $|R| \geq \frac{k l^{2/3}}{10} $ (\cref{fact:size_R}) and $|Q_1| \leq \frac{k^2 l^{4/3}}{1000}$, it holds that $|Q_1\cap (R\times R)| \leq |R|^2 / 1000$. 
%	Second, $G_2$ implies a bound on $|R'|$: indeed, fix any item $x \in R \setminus R'$, with $u_x = i$. Then there is some $y$ with $u_y = i$, such that $(x,y) \in Q_1 \cap R_i \times R_i$. 
%	Hence, $|R \setminus R'| \leq |Q_1 \cap \cup_{i=1}^k (R_i \times R_i)|$, and $G_2$ yields  $|R'|\geq 9|R|/10$. 
%	
	Hence,
	\[
	\mathbb{P}(G_3\mid G_2,G_1) = \frac{|(R'\times R')\setminus Q_1|}{|R|^2} \geq \frac{|R'|^2}{|R|^2}-\frac{|Q_1|}{|R|^2} = 1 - o(1)
	\]
	since $|Q_1| = o(k^2l^{4/3})$ and so is $o(|R|^2)$.
%	\begin{align*}
%		\mathbb{P}(G_3\mid G_2,G_1) &= \frac{|(R'\times R')\setminus Q_1|}{|R|^2}\\
%		&\geq \frac{|R'|^2}{|R|^2}-\frac{|Q_1\cap(R\times R)|}{|R|^2}\\
%		&\geq \frac{64}{81} - \frac{1}{1000} > 0.78
%	\end{align*}
%	where the last inequality follows conditioned on $G_1$ and the upper bound on $|Q_1|$. \deepc{constants need to be fixed...however, i think we should revert back to $o(1)$.}
	Since $\Pr(G_1), \Pr(G_2) = 1 - o(1)$, we get $\mathbb{P}(G_3) = 1 - o(1)$. \qedhere
\end{proof}

After the first round of queries of the protocol, if the events $\{G_1, G_2, G_3\}$ don't all occur, we let the protocol win. 
Otherwise, we proceed to the second round after telling the protocol the solutions to $Q_1$. 
The crucial observation is the following.
\begin{observation}
	Conditioned on $\{G_1, G_2, G_3\}$, $A(Q_1)$ is independent of $a$ and $b$.
\end{observation}
\begin{proof}
	 Indeed, by definition of $R'$, if $x\in R'$, then any query $(x,y)\in Q_1$ have $u_x \neq u_y$. Unless $(x,y) = (a,b)$, the answer to this query 
	 is NO. Since $G_3$ enforces $(a, b) \in (R' \times R')\setminus Q_1$, every answer in $A(Q_1)$ is determined solely by $\{u_x,r_x\}$ which are independent
	 of $a,b$.
	 %it therefore means that the answers of the queries on all items $x$ that could have $a_x = 1$ or $b_x = 1$ are ``no" independently of $a_x$ and $b_x$. 
\end{proof}

Now we can establish the upper bound in Equation~\eqref{eq:nts} completing the proof of~\Cref{thm:det}, and thus establishing the lower bound of~\Cref{res:1}.
\begin{lemma}
	$\Pr[(a,b) \in Q_2~|~A(Q_1)] = o(1)$. %\deepc{i didn't calculate constant.}
\end{lemma}
\begin{proof}
To do so, we first note
\[
\Pr[(a,b) \in Q_2~|~A(Q_1)] \leq \Pr[(a,b)\in Q_2~|~A(Q_1), G_1, G_2, G_3] + \underbrace{\sum_{i=1}^3 \Pr[\neg G_i]}_{=o(1)}
\]
By the above observation, we have the first term in the RHS is $\Pr[(a,b)\in Q_2~|~G_1,G_2,G_3]$.
Conditioned on $G_3$, any $(x,y) \in Q_2$ which is not in $(R'\times R')\setminus Q_1$ has $0$ chance of being $(a,b)$. Therefore,
we can upper bound
\[
\Pr[(a,b)\in Q_2~|~G_1,G_2,G_3] \leq \sum_{(x,y)\in Q_2} \Pr[(a,b) = (x,y)~|~(x,y)\in (R'\times R')\setminus Q_1, G_1,G_2,G_3]
\]
Event $G_3$ ensures that $(a,b)$ is uniform in $(R'\times R') \setminus Q_1$ instead of uniform in $R \times R$, and so the above probability is 
$\frac{1}{|(R'\times R')\setminus Q_1|}$. 
Conditioned on $G_1, G_2$, we have $|(R'\times R') \setminus Q_1| \geq (|R'|)^2 - |Q_1| = (1 -o(1))k^2l^{4/3}$, implying each term in the above summation 
is $\frac{1+o(1)}{k^2l^{4/3}}$. If $|Q_2| = o(k^2l^{4/3})$, the above summation is $o(1)$.
%
%
%\geq \frac{64}{81} \cdot \frac{81 k^2 l^{4/3}}{100} - \frac{k^2 l^{4/3}}{1000} \geq \frac{6}{10} \cdot k^2 l^{4/3}$, we get that each term is $\leq \frac{10}{6}\cdot \frac{1}{k^2l^{4/3}}$.
%If $|Q_2| \leq \frac{n^{4/3}k^{2/3}}{1000}$, then the above is $\leq o(1)$. 
\end{proof}

\subsection{A $3$-round lower-bound for unknown $k$}\label{sec:lower-bound-unknown-k}

The main goal of this section is to show that there cannot exist a $3$-round randomized algorithm which, with high probability, learns any partition $\calP$ 
making $\Ot(n|\calP|)$ queries per round. To do so we create a distribution over partitions, and with $0.5$ probability the partition is the trivial partition with 
$1$ part. What this forces is that any randomized algorithm must make $\Ot(n^{4/3})$-queries in the first round; otherwise, with $0.5$ probability, the number of queries overshoots the budget.
This is all we need to prove our $3$-round lower bound; we will in fact tell the algorithm what $|\calP|$ is along with the first round's answers.
As we mention in~\Cref{res:2}, we show something stronger; if there is a $3$-round randomized algorithm, then it must make $\Omega(n^{4/3}|P|^{2/3})$ queries. However, we do not know of such an algorithm (and we leave that as an open problem).

\begin{theorem}\label{thm:unknown-lb}\label{thm:lower-bound-unknown-k}
	If there exists a randomized $3$-round protocol which succeeds with probability at least $\frac{9}{10}$ to find any partition $\calP$, 
	then it must make $\Omega(n^{4/3}|\calP|^{2/3})$ queries on some partition $\calP$. 
\end{theorem}

%\david{An idea to reduce the assumption on $k$. Perhaps we can keep the instance and say : in the first round, any yes answer just reduces the size of the instance by $1$, and there's $o(n)$ many yes. Then we "remove" high degree vertices, $o(n)$ of them, and consider an independent set in the remaining graph. The argument we're making about $Q_2$ should work in this remaining graph. However, we may not be able to use the 2 round case proof as the high degree vertices can still be used by the algorithm.}

This distribution $\calD'$ that we use for Yao's lemma is as follows:
with probability $1/2$, the instance is a trivial partition into a single part. With the remaining probability, it is the $2$-round distribution from the previous section with for $k = \omega(n^{2/3})$. 
%\david{I don't know the exact value of $k$ we should use: just that it has to be larger than $\sqrt{n}$ aditi: with $k>>n^{2/3}$, you can get away with making $o(n^{4/3})$ queries and the rest of the proof should be similar. } 

\paragraph{Proof.}
Fix any deterministic algorithm $\calA$ which makes $o(n^{4/3}|\calP|^{2/3})$-many queries on any partition $\calP$ drawn from the above distribution. 
Let $Q_1$ be the first round queries. Since with $0.5$ the partition may be the trivial partition (and $\calA$ doesn't know this before Round 1), we must have $|Q_1| = o(n^{4/3})$.
Note that $Q_1$ is the same even for the ``hard partitions'', and indeed, henceforth we assume that the instance is the hard $2$-round instance (if not, we let the algorithm win, which happens with $0.5$ probability).
%In the first round, the algorithm is allowed only $O(n^{4/3})$ queries, as otherwise it would exceed the allowed budget on the trivial instance. 
%We condition from now on on the instance being the hard $2$-rounds instance.

We say an item $x \in U$ is ``revealed" after the first round of queries $Q_1$ if either there is a query $(x,y) \in Q_1$ with $u_x = u_y$, or if $|\{y: (x, y)\in Q_1\}| \geq k/3$.
Any $i\in [k]$ is ``revealed" if $U_i$ contains one item that is revealed. 
We define $U' := \{x \in U: x \text{ is not revealed}\}$.

\begin{claim}
    If $|Q_1| = o(n^{4/3})$, then with probability $1-o(1)$ there are at most $o(k)$ many revealed $i$, and $|U'| = \Omega(n)$.
\end{claim}
\begin{proof}
The number of vertices with $|\{y: (x, y)\in Q_1\}| \geq k/3$ must be less than $3|Q_1|/k$. Since $|Q_1|=o(n^{4/3})$ and $k=\omega(n^{2/3})$, this is $o(k)$. Thus, there are $o(k)$ parts revealed because they contain such a vertex. 
%    First, the number of vertices with  $|\{y: (x, y)\in Q_1\}| \geq k/3$ must be less than $3|Q_1|/k$: since $|Q_1| = O(n) = O(kl)$ and $k = n^{2/3}$, this is $O(l) = o(k)$. Hence, there are $o(k)$ parts revealed because they contain such a vertex.

For any $(x,y)\in Q_1$, $\Pr(u_x = u_y) = 1/k$. Hence, the expected number of revealed parts because they contain two items of a query is $|Q_1| / k = o(k)$ as well. %Since $n = kl$ and $k = n^{2/3}$, $l = o(k)$. 
Thus, Markov's inequality ensures that, with probability $1-o(1)$, the number of revealed part is $o(k)$.
\end{proof}

After the first round, the partition of points in all revealed $i$ is given to the algorithm. In addition, $k$ is given to the algorithm.
Let $Q_2$ be the second round of queries. We assume therefore that all those are within $U'$. Our claim is essentially that this instance is the same as the two round instance with known $k$. However, we need to adjust the proofs to account for the results of queries in $Q_1$ -- which the algorithm could possibly use to gain information.

Recall~\Cref{fact:size_R} that $|R| = (1 - o(1))kl^{2/3}$ with probability $1-o(1)$, and this even is called $G_1$.

\begin{observation}
    Conditioned on $G_1$, it holds that that $|Q_2\cap (R\times R)|\leq |Q_2|=o(|R|^2)$. 
\end{observation}
\begin{proof}
	This is simply because $|Q_2| = o(n^{4/3}k^{2/3}) = o(k^2l^{2/3})$.
\end{proof}
%\begin{lemma}
%We have $\mathbb{E}[|Q_2\cap (R\times R)| | U', A(Q_1)]=o(k^2 l^{2/3})$. Additionally, with probability $1-o(1)$, conditioned on $U'$ and $A(Q_1)$, $|Q_2\cap (R\times R)|=o(k^2 l^{2/3})$. We denote this event by $G_1$
%\end{lemma}
%\begin{proof}
%Fix a query $(x, y) \in Q_2$, with $x, y \in U'$. The first round queries related to $x$ and $y$ are NO queries only based on the value of $u_x, u_y$: in particular, they are independent of $r_x$ and $r_y$. We can therefore reason as if $r_x$ and $r_y$ are sampled only in the second round. Hence, $\Pr(r_x = r_y = 1 | A(Q_1), x,y\in U') = 1/l^{2/3}$, as $r_x$ and $r_y$ are independently equal to $1$ with probability $1/l^{1/3}$. 

%By linearity of expectation, we can therefore conclude that:
% $\mathbb{E}[|Q_2\cap (R\times R)| | A(Q_1), U'] = |Q_2| \cdot 1/l^{2/3} = o(n^{4/3}k^{2/3}) \cdot 1/l^{2/3} = o(k^2 l^{2/3})$. 

%\end{proof}

\begin{lemma}\label{lem:unknown-b2}
    We have $\mathbb{E}[|Q_2\cap \cup_{i=1}^k(R_i\times R_i)| \mid A(Q_1)]=o(k l^{2/3})$. With probability $1-o(1)$, conditioned on $A(Q_1)$, $|Q_2\cap \cup_{i=1}^k(R_i\times R_i)|=o(k l^{2/3})$. We refer to this event as $G_2$.  
\end{lemma}
\begin{proof}
Fix a query $(x, y) \in Q_2$, and consider the randomness remaining after round 1 in the choice of $u_x$ and $u_y$. 
As $x,y\in U'$, they are chosen uniformly at random in $[k]$, conditioned on $\forall (x, z) \in Q_1, u_x \neq u_z$ and $\forall (y, z) \in Q_1, u_y \neq u_z$. Suppose that all such $u_z$ are given, and let $k_x := |\{u_z: (x,z) \in Q_1\}|, k_y := |\{u_z: (y,z) \in Q_1\}|$ and $k_{xy} := |\{u_z: (x,z) \in Q_1\}  \cap \{u_z: (y,z) \in Q_1\} |$. 
We therefore have:
\[\Pr(u_x = u_y | A(Q_1), k_x, k_y, k_{xy}) = \frac{k - k_x - k_y + k_{xy}}{(k-k_x)(k-k_y)}.\]
% As $x,y\in U'$, $u_x$ and $u_y$ are chosen uniformly at random in $[k]$, conditioned on $\forall (x, z) \in Q_1, u_x \neq u_z$ and $\forall (y, z) \in Q_1, u_y \neq u_z$. 
Now, conditioning on $x$ and $y$ being not revealed implies $k_x, k_y \leq k/3$. Thus, in this case the probability becomes:
% there are at most $k/3$ such $z$ for both $x$ and $y$. Suppose that all such $u_z$ are given, and let $k_x := |\{u_z: (x,z) \in Q_1\}|, k_y := |\{u_z: (y,z) \in Q_1\}|$ and $k_{xy} := |\{u_z: (x,z) \in Q_1\}  \cap \{u_z: (y,z) \in Q_1\} |$. 
\begin{align*}
    \Pr(u_x = u_y | A(Q_1), x,y \in U') \leq \frac{k}{(k-k/3)^2} = \frac{9}{4k}.
\end{align*}
%\aditi{here conditioning on $x,y\in U'$ already implies the bounds on $k_x,k_y,k_{xy}$, right?}\david{yes, this is what wanted to say, is this clearer now?}
% , even without the conditioning on $k_x, k_y, k_{xy}$, we have $\Pr(u_x = u_y | A(Q_1), x,y\in U') \leq \frac{9}{4k}$.

    In addition, as in the previous proof, $u_x, u_y$ are independent from $r_x, r_y$ and $\Pr(r_x = r_y = 1 | A(Q_1), x,y\in U') = l^{-2/3}$. Hence,
    $\Pr((x,y) \in \cup_{i=1}^k(R_i\times R_i)) \leq \frac{9}{4k} \cdot l^{-2/3}$, and linearity of expectation ensures that:
    $$\mathbb{E}[|Q_2\cap \cup_{i=1}^k(R_i\times R_i)| | A(Q_1), U'] \leq |Q_2| \cdot \frac{9}{4k} \cdot l^{-2/3} = o(k^2 l^{4/3}) \cdot \frac{9}{4k} \cdot l^{-2/3} = o(k l^{2/3}).$$

    Markov inequality concludes.
\end{proof}

\begin{lemma}\label{lem:unknown-b3}
    Let $R'=\{x\in R \cap U'\mid \text{for all }(x,y)\in Q_2, u_x\neq u_y\}$, and $G_3$ be the event that $(a,b)\in (R'\times R')\setminus Q_2$.
Then, $\Pr(G_2, G_3 | A(Q_1)) = 1-o(1).$
\end{lemma}
\begin{proof}
    This proof follows the same line as that of \Cref{lem:b3}: first, we note that
    $|R \setminus R'| \leq |Q_2 \cap \cup_{i=1}^k(R_i\times R_i)| $. Conditioned on $A(Q_1)$ and $G_2$, it therefore holds that $|R \setminus R'| = o(k l^{2/3})$. Since  $|R| = (1+o(1))k l^{2/3}$, this implies that, conditioned on $G_2$, $R' =(1 + o(1))k l^{2/3}$, and hence $|R' \times R' \setminus Q_2| = (1+o(1))k^2 l^{4/3} = (1+o(1))|R\times R|$.

    The pair $(a, b)$ is initially chosen uniformly in $R \times R$: as, conditioned on $G_2$, $R' \times R' \setminus Q_2$ takes a $1-o(1)$ proportion of this set, $(a,b)$ lies in it with probability $1-o(1)$. 
    We conclude the proof using that $G_2$ occurs with probability $1-o(1)$ (\Cref{lem:unknown-b2}).
\end{proof}

\begin{lemma}\label{lem:ab-uniform}
    For any $x,y$, $\mathbb{P}(x=a,y=b\mid x, y \in R', A(Q_1), A(Q_2), G_1,G_2,G_3) \geq \frac{1+o(1)}{k^2l^{4/3}}$
\end{lemma}
\begin{proof}
    Just as in the known-$k$ case, the crucial observation is that, once conditioned on $\{G_1, G_2, G_3\}$, the values of $a$ and $b$ are independent of the answers to the first two rounds queries -- and therefore, $G_3$ ensures that the pair $(a, b)$ is uniform in $(R'\times R')\setminus Q_2$. As shown in the proof of \Cref{lem:unknown-b3}, $|(R'\times R')\setminus Q_2| = (1+o(1))k^2 l^{4/3}$ with probability $1-o(1)$. Since $G_2$ and $G_3$ occur with probability $1-o(1)$, this concludes the lemma.
\end{proof}

Hence, after the first two rounds of queries, the pair $(a,b)$ is almost uniform in a set of size $k^2l^{4/3}$. Just as in the known-$k$ case, if the third round uses significantly less queries, it therefore cannot hit the pair $(a,b)$ and the algorithm cannot distinguish between SINGLETONS and DOUBLETON with all but negligible probability. We repress the repetition of this calculation.
This concludes the proof of \Cref{thm:unknown-lb}.

%\newpage
\section{Deterministic algorithms with unknown number of parts } \label{sec:det-alg-unknown-k}

Black, Mazumdar, and Saha~\cite{BMS25} studied the rounds-vs-query trade-off for learning a partition $\calP$ when $|\calP|$ is known, and give a complete answer for deterministic algorithms.
As mentioned in the introduction, the similar trade-off question for the unknown $|\calP|$ case, to the best of our knowledge, was not studied even for deterministic algorithms.
In this section, we show an {\em exponential} gap in the number of rounds needed by deterministic algorithms of (near) optimal query-complexity, between the known and unknown $|\calP|$ cass.
%
%
%currently we do not have a full understanding of round-versus-query complexity trade-offs even for deterministic algorithms (unlike the case of known $|\calP|$, as illustrated in~\cite{BMS25}) -- while the two previous sections draw a clean picture of the landscape for randomized algorithms. Below, we make a few observations and leave some questions open. Throughout this section, $\calP$ is some hidden partition and algorithms do not have 
%any idea about $|\calP|$ (other than it is an integer between $1$ and $n$).

\subsection{Algorithms.} 

We begin with a simple algorithm akin to the deterministic algorithm in BMS~\cite{BMS25}.
\begin{theorem}\label{thm:upper-bound-det-unknown-k}
	For any $r \geq 1$, there is a deterministic algorithm $\calA$ which makes $O(n^{1 + \frac{1}{r}}|\calP|)$-many PAIR queries in $r$-rounds to learn any partition $\calP$. 	
	The algorithm doesn't need to know $|\calP|$.
\end{theorem}
\begin{proof}
	Let $k:=|\mathcal P|$ and let $q:=\lceil n^{1/r}\rceil$; we assume $q\geq 2$.
	Note that the algorithm does not know $k$.  Fix an arbitrary ordering of the
	elements of the universe.  The algorithm maintains a partition of the
	universe into consecutive blocks, together with the restriction of
	$\calP$ inside each block. Initially each block is a singleton.
	
	In
	round $t$, we group $q$ consecutive current blocks into a ``parent block''; each of these consecutive blocks are ``children blocks'' of this parent block.  In
	each parent block, we take one representative from every currently known
	local part of its children blocks, and query all pairs of these representatives.  Two local
	parts are merged if and only if the corresponding representative pair answers YES.
	This recovers the exact restriction of $\mathcal P$ inside the parent block.
	Note that after $t$ rounds, each parent block has size $q^t$, and so, since
	$q^r\ge n$, after $r$ rounds there is only one parent block, and therefore, the algorithm learns $\calP$ at the end of the $r$th round.
	
	It remains to bound the number of queries.  Consider round $t$ in which
	the children block size is $q^{t-1}$.  Each children block contains at most
	$\min\{q^{t-1},k\}$ local parts; here $k:=|\calP|$ which the algorithm doesn't know, but this is only for analysis.  
	Hence a parent block contains at most $q\min\{q^{t-1},k\}$ local parts, and so the number of PAIR queries inside one
	parent block is
	$O\left(q^2\min\{q^{t-1},k\}^2\right)$.
	There are $O(n/q^t)$ parent blocks, so the total number of queries in
	this $t$th round is
	\[
	O\left(
	\frac{n}{q^t}\cdot q^2\min\{q^{t-1},k\}^2
	\right)
	=
	O\left(
	nq\cdot \frac{\min\{q^{t-1},k\}^2}{q^{t-1}}
	\right).
	\]
	Summing over $1\leq t\leq r$, we get the total number of queries made is $O(nq)$ times
	\[
	\sum_{t=1}^r  \frac{\min\{q^{t-1},k\}^2}{q^{t-1}} 
	\leq 
	\underbrace{\sum_{t: q^{t-1}\le k} q^{t-1}}_{\leq 2k}
	+
	\underbrace{\sum_{t: q^{t-1} > k} \frac{k^2}{q^{t-1}}}_{< k\sum_{j=0}^\infty \frac{1}{q^j} \leq 2k}
	=
	O(k).
	\]
	where the summations follow because $q \geq 2$.
	Therefore the total number of queries is $O(nqk)$ proving the theorem's assertion.
\end{proof}

\begin{corollary}
	There is an $\ceil{\log n/\log \log n}$-round deterministic algorithm which learns any hidden partition $\calP$ making $\Ot(n|\calP|)$ many queries.
\end{corollary}
The above corollary follows by setting $r = \ceil{\frac{\log n}{\log\log n}}$ so that $n^{1/r} = O(\log n)$; the total query complexity is $O(n|\calP|\log^2 n/\log\log n)$.
The theorem also implies that in $3$-rounds, we have a $O(n^{4/3}|\calP|)$-query deterministic algorithm; compare this with the $\Omega(n^{4/3}|\calP|^{2/3})$-lower bound for {\em randomized} algorithms we proved in~\Cref{thm:lower-bound-unknown-k}. There is a ``$|\calP|^{1/3}$-gap'', and we leave plugging this as an open question.

It is instructive to remind the reader that when $|\calP|$ is known to the algorithm, then BMS~\cite{BMS25} also consider more nuanced query complexities of the form $O(n^{1+c}|\calP|^{1-c})$, where the exponents of $n$ and $|\calP|$ add up to $2$. These are not the optimal $n|\calP|$, but it gives a nice trade-off. BMS~\cite{BMS25} showed that for arbitrary small constant $c > 0$, they can achieve {\em constant}-round algorithms with this query complexity, where the constant depends on $c$ but is independent of $n$. 

In the case when $|\calP|$ is unknown, we show that our above algorithm can be modified slightly to give an improvement in the round-complexity; in particular, we can bring down the round-complexity from nearly logarithmic to $O_c(\log\log n)$. On the other hand, our lower bound (\Cref{thm:lower-bound-det-unknown-k-nuanced}) will show $o(\log\log n)$ rounds is not possible for a nuanced query complexity of $n^{1+c}|\calP|^{1-c}$ for {\em any} $c < 1$; of course, when $c=1$, there is a trivial non-adaptive $O(n^2)$-query algorithm which makes all pairwise queries. 

\begin{theorem}\label{thm:upper-bound-det-unknown-k-nuanced}
For any $c > 0$, there is an $O(\frac{1}{c}\cdot \log\log n)$-round algorithm to learn $|\calP|$ that makes $O(n^{1+c}|\calP|^{1-c})$ many queries.	
\end{theorem}
%Later we show that if the number of rounds is $o(\log\log n)$, then one cannot get any $n^{1+c}|\calP|^{1-c}$-query algorithms for any constant $c < 1$. In that sense, the above upper-bound is tight.
\begin{proof}
	As in the proof of~\Cref{thm:upper-bound-det-unknown-k}, we maintain 
	a partition of the
	universe into consecutive blocks, together with the exact restriction of
	$\calP$ inside each block. Earlier, the block sizes grew by a factor of $q = n^{1/r}$. This time, 
	the block sizes are governed as follows:
	\[
	b_0=1,\qquad b_{t}=n^c b_{t-1}^{1-c},
	\]
	rounding as necessary. Note that this implies
	\[
	b_t=n^{1-(1-c)^t},
	\]
	and so when $t = \frac{1}{c}\cdot \log\log n$, this becomes $n$.
	
	In round $t$, we merge $\ceil{b_{t}/b_{t-1}}$ child blocks into one parent
	block; as before, the algorithm takes one representative from every
	known local part of the children block in each parent block, and queries all pairs of these
	representatives.  It then merges two local parts iff their representatives
	answer YES. 
	This recovers the exact restriction of $\calP$ in $O(\frac{1}{c}\cdot \log\log n)$-rounds.
	the parent block, so the invariant follows by induction. What is interesting is the number of queries.
	
	It remains to bound the number of queries.  In round $t$, there are $O(n/b_t)$ many parent blocks.
	For each such parent block, we have $\ceil{\frac{b_t}{b_{t-1}}}$ children, and each such child contributes $\leq \min(b_{t-1}, k)$ local parts.
	Again, comparing with the previous proof, $\frac{b_t}{b_{t-1}}$ was $q$ and $b_{t-1} = q^{t-1}$.
	This time, the total number of queries made in round $t$ is
	\[
	O\left(\frac{n}{b_t} \cdot \left(\frac{b_t}{b_{t-1}} \min(b_{t-1}, k)\right)^2\right) = O\left(nb_t \cdot \frac{\min(b_{t-1}, k)^2}{b_{t-1}^2}\right)
	\]
	By plugging in our choice of $b_t = n^c b^{1-c}_{t-1}$, we get that the RHS is $n^{1+c}b_{t-1}^{1-c}$ if $b_{t-1} \leq k$, and $n^{1+c} \frac{k^2}{b^{1+c}_{t-1}}$ otherwise.
	Summing up over all $t$, we get that the number of queries is $O(n^{1+c})$ multiplied by
	\[
	\sum_{t: b_{t-1} \leq k} b_{t-1}^{1-c} ~~+~~ \sum_{t: b_{t-1} > k} \frac{k^2}{b_{t-1}^{1+c}}
	\]
	We claim that this is $O(k^{1-c})$ which would prove the lemma. Indeed, the $b_t$'s growth is faster than a geometric progression since $\frac{b_{t+1}}{b_t} = \left(n/b_t\right)^c > 1$.
	Thus, for the first term, the largest $b_t$ matters and for the second summation, the smallest one matters; and both give $k^{1-c}$.
\end{proof}
\subsection{Lower Bounds.} 
We begin with a lower bound which strongly complements~\Cref{thm:upper-bound-det-unknown-k-nuanced}; this acts as a warm-up to the $\Omega(\log n/\log\log n)$-lower bound for deterministic algorithms
with near optimal query complexity. We show that unless one makes $\log\log n$-rounds of queries, one cannot get even an $n^{1.999}|\calP|^{0.001}$-query algorithms; while the aforementioned theorem
states that with $O(\log\log n)$ rounds, we get $n^{1.001}|\calP|^{0.999}$-query algorithms.
%
% getting the aforementioned cube-root $|\calP|$ saving. Next, we prove a lower bound showing that when $|\calP|$ is unknown, this query complexity is impossible for {\em deterministic} algorithms. This differentiates the the known $|\calP|$ and unknown $|\calP|$ case. Indeed, we prove something much stronger. We prove that for any algorithm with query complexity a function of $n$ and $|\calP|$ where their exponents sum to $2$, must either have a {\em quadratic} dependency on $n$, or must run in {\em super-constant} rounds. 

\begin{theorem}\label{thm:lower-bound-det-unknown-k-nuanced}
	Fix any $0\le c<1$.  There is no deterministic
	$o(\log\log n)$-round algorithm which, for every partition $\calP$ of an
	$n$-element universe, learns $\calP$ using at most
	$O(n^{1+c}|\calP|^{1-c})$-many
	PAIR queries.
\end{theorem}

Our proof below uses the adversary method as also used by BMS~\cite{BMS25} for the known $|\calP|$ case. 
However, we deviate a bit from their method. While they use the existence of large independent sets in sparse graphs as their main combinatorial tool, we use
that sparse graphs have low chromatic number. The latter is relevant for our lower bound since we use low chromatic number to indicate the inability of the algorithm to deduce ``large $|\calP|$''. Details follow.

\begin{proof}
	For the sake of contradiction, suppose there is an $r$-round algorithm with $r = o(\log\log n)$.
	
	The main idea is this. 
	We, as the adversary, will keep answering NO to {\em every}  PAIR query the algorithm asks. 
	The trivial partition $\calP_0$ where every part is a singleton is consistent with this answer. 
	Furthermore, at the end of the $r$th round, if there is some pair $(a,b)$ which the algorithm has {\em not} queried
	sometime, then it cannot be correct: it cannot distinguish from $\calP_0$ and the partition which clubs $\{a,b\}$ together in one part, and the remaining parts are singletons.
	In particular, if $Q_1, Q_2, \ldots, Q_r$ are the queries made by the algorithm (where $Q_i$ could depend on the answers to $Q_1, \ldots, Q_{i-1}$), 
	the union of this must be all the edges. Or, $\sum_{i=1}^r |Q_t| = \Omega(n^2)$.
	
	At the beginning of round $i$, however, the algorithm's budget is $O(n^{1+c}k_i^{1-c})$ where $k_i$ is the largest lower bound the algorithm can infer from all the NO answers it has received so far. To this end, let $F_{i} := Q_1 \cup \cdots \cup Q_{i}$ with $F_0 = \emptyset$, and let $G_i = (U,F_i)$. 
	Then note that 
	\[
	k_i = \chi(G_{i-1})
	\]
	where $\chi(G)$ is the chromatic number of a graph $G$.
	This is almost by definition: any valid coloring gives a partition which is consistent with the NO answers so far, and every consistent partition also induces a coloring. 
	
	Now, if number of queries are small, or equivalently if $G_{i-1}$ is not dense, then we can upper bound its chromatic number. 
	More precisely, we know that every $m$-edge graph\footnote{Indeed, $\chi(G) \le 1 + \sqrt{2m}$; see Fact 1 of~\cite{AC08}.} $G$ has $\chi(G) = O(\sqrt{m})$. Now we can set up a recurrence which will lead us to the required lower bound; we do so next.
	
	\begin{claim}
		Let $a_i = 	1-\left(\frac{1-c}{2}\right)^i$. Then, $\chi(G_i) = O(n^{a_i})$.
	\end{claim}
	
	\begin{proof}
		The proof is by induction. When $i = 0$, we have $a_0 = 0$ and $G_0$ is the empty graph with $\chi(G_0) = 1$.
		Suppose the claim is true for some $i$ and we wish to prove it for $i+1$. 
		The queries made by the algorithm till and including round $i$ has an upper bound of $O(n^{1+c}k^{1-c}_i)$. 
		Using the fact that $k_i = \chi(G_i) = O(n^{a_i})$, we get
		\[
		|F_i| \leq O(n^{1 + c} \cdot n^{a_i(1 - c)}) ~~\underbrace{\Rightarrow}_{\chi(G) = O(\sqrt{|E(G)|})}~~ \chi(G_i) = O\left(n^{\frac{(1+c) + a_i(1-c)}{2}}\right)
		\]
		By choice of the sequence $a_i$, we have $a_{i+1}
		=
		\frac{1+c+(1-c)a_i}{2}$.	
		This proves the claim.
	\end{proof}
	When $c < 1$ and $r = o(\log\log n)$, we have $a_{r} < 1$. Therefore, at the end of $r$ rounds, the union of all queries is a graph $G_r$ with $\chi(G_r) = o(n)$.
	In particular, it is not the complete graph. As argued in the first paragraph of this proof, the algorithm cannot be correct. This completes the proof of the theorem.
\end{proof}
Although the above lower bound rules out any $o(\log\log n)$-round deterministic algorithm making $\Ot(n^{1+c}|\calP|^{1-c})$ many\footnote{here the $\widetilde{}$ is suppressing logarithmic factors in the denominator} queries for any $c < 1$, for the $c=0$ case, unfortunately, we do not obtain a stronger lower bound; in particular, an $O(\log\log n)$-round $\Ot(n|\calP|)$-query deterministic algorithm is not ruled out. Furthermore, our scheme above has a $\log\log n$ bottleneck; in particular, if the adversary keeps answering NO to all queries (as in our construction above) the algorithm makes (possibly because the final partition is the trivial one), then an algorithm could query edges of a clique with $n^{1/2}$ vertices in round 1 and know $|\calP| \ge n^{1/2}$, and then query edges of a clique of $n^{3/4}$ vertices in round 2 and know $|\calP|\geq n^{3/4}$, query edges of a clique of $n^{7/8}$ vertices, and so on, succeeding in certifying the trivial partition in $O(\log\log n)$ rounds. 

Henceforth, to prove a stronger lower bound, our adversary needs to answer YES once in a while. We do so in the next result.

\begin{theorem}\label{thm:lower-bound-det-unknown-k}
	Any deterministic algorithm which learns any hidden partition $\calP$ using $O(n|\calP|\log^c n)$-many queries, for some constant $c \geq 0$, 
	must proceed in $\Omega(\log n/\log\log n)$ rounds.
\end{theorem}
\begin{proof}
	Let us fix a deterministic algorithm $\calA$ which makes queries $Q_1, \ldots, Q_r$ in the $r$-rounds and for contradictions sake suppose $r = o(\log n/\log\log n)$. 
	We have $|Q_i| \leq nk_i\log^c n$ for each $i$, where $k_i$ is the {\em largest lower bound} on $|\calP|$ which the algorithm can deduce given the answers seen so far.
	We reach a contradiction by showing that at the end of $r$ rounds, the algorithm cannot distinguish between two distinct partitions.
	
	The proof proceeds via an adversary argument where the adversary decides the answers to $Q_i$ after seeing it. To do so,
	the adversary maintains graphs $G_1, \ldots, G_r$ where $G_i = (V_i, F_i)$, where each vertex of $V_i$ is a subset of $U$ that the adversary has revealed belong to the same part; in particular, they are the connected components of the YES answers the adversary has given to queries in rounds $1$ to $i-1$. 
	Each edge $e\in F_i$ between two parts $A\in V_i$ and $B\in V_i$ corresponds to a query $(x,y) \in Q_j$ for $j \leq i-1$ with $x\in A$, $y\in B$, and which has received answer NO from the adversary. To begin with, $V(G_1) = U$ and $E(G_1) = \emptyset$.
	
	Let $\chi(G_i)$ be the chromatic number of $G_i$. By definition of $G_i$, there is a partition of size $k_i$ consistent with the answers given by the adversary till now. Therefore, we have that 
	\begin{equation}\label{eq:007}
		|F_i| + |Q_i| \leq n\cdot \chi(G_i)\log^c n
	\end{equation}
	The main idea now is to answer the queries to $Q_i$ in such a way that (a) the YES answers do not contradict older NO answers, (b) the YES answers do not contract $V_i$ ``too much''; in particular, we assert $|V_{i+1}| \geq |V_i|/2$, and (c) the resulting graph $G_{i+1}$ doesn't have its chromatic number blown up; in particular, we want $\chi(G_{i+1}) \leq O\left(\chi(G_i)\log^c n\right)$.
	If we manage to do so, then note that after $r$-rounds for $r = o(\log n/\log\log n)$, the adversary will have a graph $G_{r+1}$ with $|V(G_{r+1})| \geq n/2$ and $\chi(G_{r+1}) = n^{o(1)}$.
	In particular, there is at least one (in fact many) pair $(A,B) \in V_{r+1} \times V_{r+1}$ such that all answers are consistent with the two partitions: one which is $V_{r+1}$ and the other 
	is $V_{r+1}$ with $A$ and $B$ merged.
	
	To obtain (a), (b), and (c), we need the following simple lemma.
		
	\begin{lemma}
		Let $H=(V,E)$ be a graph with a proper $L$-coloring $\phi$. Let $Q$ be another collection of pairs in $V\times V$.
		For every number $A\ge 1$, there exists a coloring $\psi$ of $V$ using at most $L(1+A)$ colors such that: (i) for every $(x,y)\in E$, $\psi(x) \neq \psi(y)$, 
		(ii) contracting all pairs $(x,y) \in Q$ with $\psi(x) = \psi(y)$, decreases the number of vertices by at most $\frac{2|E| + |Q|}{AL}$.
	\end{lemma}
	
	\begin{proof}
	This follows from a simple application of the probabilistic method with alteration.
	Note that $\phi$ uses $L$ colors. Use a fresh palette of $AL$ colors, and 
	randomly color every vertex independently using this. Call this random coloring $\rho$. 
	
	Let $X$ denote the number of $(u,v)\in E$ with $\rho(u) = \rho(v)$, and let $Y$ denote the number of $(u,v)\in Q$ with $\rho(u) = \rho(v)$.
	Note that $\Exp[X] = \frac{|E|}{AL}$ and $\Exp[Y] = \frac{|Q|}{AL}$. To obtain $\psi$, for each edge in $E$ which is monochromatic, revert its original color.
	More precisely, if we let $B$ denote the set of vertices incident to the $X$ edges of $E$ which are
	monochromatic under $\rho$, then 
		$\psi(v)= \phi(v)$ if $v\in B$ and $\psi(v) = \rho(v)$ if $v\notin B$.
	Note that $|B| \leq 2X$.
	
	Now, by design, we have obtained (i) since the palettes used by $\rho$ and $\phi$ are disjoint. We now need to argue (ii). 
	Consider the $\psi$-monochromatic pairs of $Q$. Such a pair of vertices is
		either already monochromatic under $\rho$, contributing to $Y$, or has
		both vertices in $B$ (if there is only one end point, then it won't be monochromatic since we use fresh colors). 
		Contracting the first type loses at most $Y$
		vertices, and contracting all edges inside $B$ loses at most
		$|B|\le 2X$ vertices. Thus the total number of vertices lost is at most $2X+Y$.
		Since $\Exp[2X + Y] \leq \frac{2|E| + |Q|}{AL}$, there exists a coloring as required by the lemma.
	\end{proof}
	Let us now return to the adversary. Inductively, we assume that at the beginning of round $i$, we have 
	\[
		|V(G_i)| \geq n\cdot\left(1 - \frac{i}{8r}\right) ~~\text{and}~~ \chi(G_i) \leq \left(16r\log^{c} n\right)^i
	\]
	Upon receiving the queries $Q_i$, the adversary sets
	$A:=8r\log^c n$, and applies the above lemma with
	\[
	H=G_i(V_i, F_i),\qquad Q=Q_i,\qquad L=\chi(G_i)
	\]
	The lemma leads to a graph $G_{i+1}$ with 
	\[|V(G_{i+1})|\geq |V(G_i)| - \frac{2|F_i| + |Q_i|}{A\chi(G_i)} \underbrace{\geq}_{\text{\eqref{eq:007}}} |V(G_i)| - \frac{2n\chi(G_i)\log^c n}{A\chi(G_i)} = |V(G_i)| - \frac{n}{8r},\] 
	and 
	\[
	\chi(G_{i+1}) \leq \chi(G_i)\cdot (1 + A) \leq 2A \chi(G_i) = (16r\log^c n)\cdot \chi(G_{i+1})
	\]
	Both imply that the inductive hypothesis holds at $i+1$ as well.

	The adversary answers YES to the $Q_i$-edges that are
	monochromatic under this coloring, answers NO to the remaining
	$Q_i$-edges, and then contracts the YES-connected components. By design, our answers are consistent with previous answers.

At the end of round $r$, we have $|V(G_{r+1})| \geq n/2$ and $\chi(G_{r+1}) \leq \left(16r \log^c n\right)^r$, and if $r = o(\log n/\log\log n)$, we have $\chi(G_{r+1}) = n^{o(1)}$.
This leads to the contradiction as described earlier.
\end{proof}

\section{Weak Subset Queries}

We now go back to the randomized algorithms, and extend them to use the power of weak subset queries in order to reduce the number of queries made to the oracle. Building on the techniques of BMS \cite{BMS25}, we show the following theorem:
\begin{theorem}\label{thm:known-subset}
	    Given an upper bound $k$ on the size of a partition $\calP$, there is a 3 round randomized algorithm for learns $\calP$ with $\Ot(\frac{nk}{s^2})$ many weak-subset queries of size at most $2\leq s \leq \sqrt{k}$.
\end{theorem}

First note that if $k < s^2$, then $\frac{nk}{s^2} \leq n$: since $\Omega(n)$ queries are information-theoretically necessary (because each query returns $\log k$ bits of information, and that there are $k^n = 2^{n \log k}$ many possible partitions \cite{BLMS24}), there cannot be an algorithm scaling with $s^2$. Hence, we require for the theorem that $s \leq \sqrt{k}$.

To prove \Cref{thm:known-subset}, we use the following result of \cite[Lemma 4.1]{BMS25}.

\begin{lemma}[Lemma 4.1 of \cite{BMS25}]\label{thm:bms}
    There is a non-adaptive algorithm which for any query size bound $2\leq s\leq \sqrt{n}$ and error probability $\delta>0$, learns an arbitrary partition on $n$ elements exactly using $O(\frac{n^2}{s^2}\log(n/\delta)+n \log^4(n/\delta)\log s)$ subset queries of size at most $s$, and succeeds with probability at least $1-\delta$. 
\end{lemma}

The algorithm from \Cref{thm:bms} can be readily used to improve rounds 1 and 3 of \Cref{alg:known}: as  $R$ and $P'$ have size $O(\sqrt{nk \log n})$ (by \Cref{lem:rem_points}), learning completely the partition in each of them requires only $\Ot\left(\max\{\sqrt{nk}, \frac{nk}{s^2}\}\right) = \Ot\left(\frac{nk}{s^2}\right)$ many queries. It is therefore merely necessary to improve the query complexity in Round 2.
The question is therefore the following: given a subset $S$ of one representative from a subset of the clusters, can we learn the whole clusters of points in $R$, in $\Ot(\frac{nk}{s^2})$ many non-adaptive queries?

\begin{lemma}\label{lem:subset-find-clusters}
	Let $A \subseteq P$ and $|P|=n$. For $y  \in A$, let $C_y$ be $y$'s cluster in $P$. It is possible to learn all $C_y$, in one round, using $\Ot(\frac{n |A|}{s^2})$ weak subset queries.
\end{lemma}
\begin{proof}
	Partition $P$ into $L_1, L_2, \ldots, L_p$ each of size $s^2$ (except for $L_p$, which has size at most $s^2$), and partition $A$ into $A_1, \ldots, A_q$ each of size $s^2$ (except for $A_q$, which has size at most $s^2$), where $p = n/s^2$ and $q = |A|/s^2$.
	
	Using the algorithm from \Cref{thm:bms}, we can in one round learn the partition of $L_i \cup A_j$ using $\Ot(2s^2/s)^2 = \Ot(s^2)$ many weak subset queries for each $i, j$. The total number of queries made to learn the partition in all $L_i \cup A_j$ simultaneously is $\Ot(pqs^2) = \Ot(n|A|/s^2)$.
	
	This information is enough to recover the cluster of all points in $A$: fix points $x \in P$ and $y \in A$. Since the $L_i$ partition $[n]$, there is $i$ such that $x \in L_i$; similarly, there is a $j$ such that $y \in A_j$. When learning the full clustering of $L_i \cup A_j$, the algorithm will learn whether $x \in C_y$ or not. Since this holds for all $x$, this concludes the lemma.
\end{proof}

Combining this lemma with the results of the pairwise queries algorithm concludes the proof of \Cref{thm:known-subset}:
\begin{proof}[Proof of \Cref{thm:known-subset}]
	As sketched above, the algorithm follows the three steps of \Cref{alg:known}; and uses the algorithm from \Cref{lem:subset-find-clusters} instead of round 2.
	
	More formally, the algorithm starts with sampling a set $R$ of $\sqrt{nk \log n}$ random points.
	In round 1, the algorithm uses \Cref{thm:bms} to learn the partition of $R$ in $\Ot(\sqrt{nk \log n}/s)^2 = \Ot(nk/s^2)$ many queries. Then, the algorithm selects in $R$ a set $S$ of at most $k$ representatives, one per cluster.
	Then, in round $2$, the algorithm uses \Cref{lem:subset-find-clusters} (using $S$ for the set $A$) to find all points in the clusters of points from $S$, in $\Ot(\frac{nk}{s^2})$ queries. This in particular implies that all points in $R$ are correctly classified.
	Finally, let $P'$ be the set of remaining unclustered points: the third round is again an application of \Cref{thm:bms} to get the full clustering of $P'$, using $O(|P'|^2 / s^2)$ queries.
	
	This algorithm has the same behavior as \Cref{alg:known}: hence, \Cref{lem:known-correct} shows its correctness. 
	The number of queries is $\Ot(nk/s^2)$ in the first two rounds, and \Cref{lem:rem_points} shows that $|P'| \leq \sqrt{nk \log n}$: hence, the number of queries in the thirds round is as well $\Ot(nk/s^2)$, which concludes the theorem.
\end{proof}

A direct corollary of this is the adaptation of the algorithm for unknown $|\calP|$:

\begin{corollary}[Weak subset queries]
	There is a $4$ round algorithm which learns a partition $\calP$ and, with high probability, makes $\Ot(\frac{nk}{s^2})$ many weak subset queries of size at most $2\leq s \leq \sqrt{k}$. This algorithm has no knowledge about $|\calP|$.
\end{corollary}
\begin{proof}
	This result is a direct corollary of \Cref{alg:unknown} and of \Cref{thm:known-subset}. To implement \Cref{alg:unknown} with weak subset queries, we proceed as follows: the first round uses $\Ot(n)$ PAIR queries (and does not use the power of the subset queries); and the subsequent rounds follow the implementation of \Cref{alg:unknown} with subset queries. 

    Since $s^2 \leq k$, we have $\Ot(n) = \Ot(\frac{nk}{s^2})$, so there is nothing to show for the first. 
	The second and third round thus use $\Ot(\frac{nk(s^*)}{s^2}) = \Ot(\frac{nk}{s^2})$ queries. The last one uses $\Ot(\frac{|U'|^2}{s^2})$ queries (with \cref{thm:bms}), where $U'$ is the set of unclustered points and verifies $|U'|^2 \leq nk\log^2 n$ (as shown in \Cref{thm:unknown-queries}). 
\end{proof}

\bibliographystyle{plain}
\bibliography{bibliography}

%\newpage
\appendix
\section{On Estimating the number of parts}\label{sec:appendix}

As mentioned in the main body, one natural approach to solve the partition learning problem when the number of parts $k := |\calP|$ is unknown is to first estimate $k$ and then apply~\Cref{alg:known}. More precisely, return a parameter $k'$ such that $k/C \leq k' \leq Ck$ for some constant $C$. Of course, one doesn't {\em need} this to learn as illustrated by~\Cref{alg:unknown}. Nonetheless, the question of estimating the number of parts is interesting in its own right: what is the optimal query complexity? what is the number of rounds required? etc.
We make a few observations. 

\begin{claim}
	For any integers $k < r$, any non-adaptive, possibly randomized algorithm making $o(n^2k/r^2)$ queries cannot distinguish between partitions with $k$ parts
	and partitions with $r$ parts.
\end{claim}
\begin{proof}
	Fix a parameter $k$; for now consider to be any positive integer smaller than the universe size $n$.
	We construct distributions $\calD^{(1)}_{r,k}$ and $\calD^{(2)}_{r,k}$ over partitions of $U$.
	\begin{itemize}
		\item First we select $R\subseteq U$ of size $r$ by picking every element with probability $r/n$ independently.
		\item Next, we partition $R$ into $k$-sets $R_1, \ldots, R_k$ where each $x\in R$ picks $i\in [k]$ with probability $1/k$.
		\item In $\calD^{(1)}_{r,k}$, we output $\calP := \{U\setminus R, R_1, \ldots, R_k\}$ with roughly $k$ parts.
		
		In $\calD^{(2)}_{r,k}$, we output $\calP := \{U\setminus R\} \cup \{\{x\}~:~x\in R\}$ with roughly $r$ parts.

	\end{itemize}
	
	Let $Q$ be any collection of PAIR queries with $q := |Q|$. We now figure out the relation between $q, r$ and $k$ such that 
	the answers to these queries cannot distinguish the two distributions. This occurs, in particular, if $Q \cap \cup_{i=1}^k (R_i\times R_i) = \emptyset$.
	Call this event $B$. The expected size of the set is $\frac{qr^2}{kn^2}$, and so, if $q = o\left(\frac{n^2k}{r^2}\right)$, the probability $B$ occurs is $1-o(1)$.
\end{proof}

\begin{corollary}
	Any non-adaptive, possibly randomized algorithm making $\Ot(n)$ queries cannot distinguish between two partitions, one with $r = \sqrt{n}$ and one with $k = n^{\eps}$ parts
	for any $\eps > 0$.
\end{corollary}
%
%Given a set $Q$ of pair queries of size $|Q| = q$, let us figure out the probability that all queries in $Q \cap (R\times R)$ evaluate to NO. 
%
%Fix a query $(x,y)$. We are interested in the event
%$(x,y) \in R\times R$ and it answers YES. 
%This probability is $\sum_{i=1}^k \Pr[x,y \in C_i]$
%which is $k \cdot \left(\frac{r}{kn}\right)^2 = \frac{r^2}{kn^2}$. So, the expected number of queries which are of this type is $\frac{qr^2}{kn^2}$, and so if $\frac{r^2}{k}\ll n^2/q$, with $1-o(1)$ probability, all queries in $Q \cap (R\times R)$ evaluate to NO. 
%If $q = n$, then it suffices to pick $k = \frac{r^2\log n}{n}$
%
%Why is this interesting? Because, if all queries in $Q\cap (R\times R)$ evaluate to NO, then $Q$ cannot distinguish between this case and the case when all the $R$ nodes are singletons. That is, it cannot distinguish between $k = r$ and $k =\frac{r^2\log n}{n}$.
%
%In particular, if we set $r = n^{0.75}$, then in one round of $n$ queries, we cannot distinguish between $k = n^{0.75}$ and $k = \sqrt{n}\log n$.
%Note that ``effective'' $k$ is also the same as $k$ in both scenarios (effective $k$ is the number of parts all but $\sqrt{n}$ points are in).
%

\end{document}